\documentclass[a4paper,USenglish]{lipics-v2019}
\graphicspath{{./graphics/}}

\title{Stabbing Pairwise Intersecting Disks by Four Points}

\author{Paz Carmi}{Department of Computer Science, Ben-Gurion University of the Negev, Israel}{carmip@cs.bgu.ac.il}{}{}
\author {Matthew J. Katz}{Department of Computer Science, Ben-Gurion University of the Negev, Israel}{matya@cs.bgu.ac.il}{}{}
\author{Pat Morin}{School of Computer Science, Carleton University, Canada}{morin@scs.carleton.ca}{}{}
\authorrunning{P. Carmi, M. Katz, P. Morin }

\Copyright{P. Carmi, M. Katz, and P. Morin}



\category{}

\relatedversion{}

\usepackage{enumitem}
\newcommand{\etal}{{et~al. }}
\definecolor{mycolor}{rgb}{0,0,1}
\newcommand{\dmin}{d_{\min}}
\newcommand{\rmin}{r_{\min}}

\newtheorem{Observation}{Observation}

\usepackage[colorinlistoftodos,prependcaption,textsize=tiny]{todonotes}

\newcommand{\old}[1]{{{}}}

\begin{document}

\maketitle

\begin{abstract}
	In their seminal work, Danzer (1956, 1986) and Stach\'{o} (1981) established that every set of pairwise intersecting disks in the plane can be stabbed by four points. However, both these proofs are non-constructive, at least in the sense that they do not seem to imply an efficient algorithm for finding the stabbing points, given such a set of disks $D$. Recently, Har-Peled \etal (2018) presented a relatively simple linear-time algorithm for finding five points that stab $D$. We present an alternative proof (and the first in English) to the assertion that four points are sufficient to stab $D$. Moreover, our proof is constructive and provides a simple linear-time algorithm for finding the stabbing points. As a warmup, we present a nearly-trivial liner-time algorithm with an elementary proof for finding five points that stab $D$.

	\old{
	Following the seminal works of Danzer (1956, 1986) and  Stach\'{o} (1965, 1981), and the recent
	result of Har-Peled \etal (2018), we study the problem of stabbing disks by points. We
	prove that any set of pairwise intersecting disks in the plane can be stabbed by four points. Our
	proof is constructive and yields a linear-time algorithm for finding the points.
	}
\end{abstract}

\newpage
\section{Introduction}
Let $D$ be a set of pairwise intersecting disks in the plane. We say that a set $S$ of points in the plane {\em stabs} $D$, if every disk in $D$ is stabbed by a point of $S$, i.e., every disk in $D$ contains a point of $S$.
The problem of stabbing disks with as few points as possible has attracted the attention of mathematicians for the past century.

The famous Helly’s theorem states that if every three disks in $D$ have a nonempty intersection, then all disks in $D$ have a nonempty intersection~\cite{Helly1923, Helly1930, Radon1921}. Thus, any point in this intersection stabs $D$. The problem becomes more interesting if $D$ contains triplets of disks with an empty intersection. In this case, Danzer (1986) \cite{Danzer1986} and Stach\'{o} (1981) \cite{Stacho1981} have shown that $D$ can be stabbed by four points. Danzer's proof uses some ideas from his first (unpublished) proof from 1956, and Stach\'{o}'s proof uses similar arguments to those in his
construction of five stabbing points from 1965 \cite{Stacho1965}. If the disks in $D$ are unit disks, then three points suffice to stab $D$ \cite{Hadwiger1955}.

As noted in the recent paper by Har-Peled \etal \cite{Har-Peled2018}, Danzer’s proof is fairly involved, and there seems to be no obvious
way to turn it into an efficient algorithm. The construction of Stach\'{o} is apparently more comprehensible, but does not seem to lead to an efficient algorithm for generating the stabbing points. Consequently, several teams of researchers have put significant and continuous effort in obtaining an efficient algorithm for generating the stabbing points, alas without success.
Har-Peled \etal \cite{Har-Peled2018} presented a simple linear-time algorithm for finding a set of five stabbing points, and left the problem of efficiently finding a set of four stabbing points as an open problem.

In this paper, we present a new proof (and the first in English) to the claim that four points are sufficient to stab $D$. Moreover, our proof is constructive and provides a linear-time algorithm for finding the stabbing points. Even though the algorithm itself is quite simple, the proof is rather involved; it is based on a careful case analysis, which is perhaps not as elegant as one would hope for, however, coming up with the right set of cases is highly non-trivial. Moreover, we use a small set of elementary geometric observations to resolve all the cases. As a warmup, we present a nearly-trivial linear-time algorithm for stabbing $D$ by five points. It's proof is elementary, and it serves as an introduction to our proof technique for the four point case.

As for lower bounds, Gr\"{u}nbaum gave an example of 21 pairwise intersecting disks that cannot be stabbed by three points \cite{Grunbaum1959}. Danzer reduced the number of disks to ten \cite{Danzer1986}, however, according to \cite{Har-Peled2018}, this bound is hard to verify. Har-Peled \etal gave a simple construction that uses 13 disks \cite{Har-Peled2018}. It would be interesting to find a simple construction for ten disks. Since every set of eight disks can be stabbed by three points \cite{Stacho1965}, it remains open to verify whether or not every set of 9 disks can be stabbed by three points.


\section{The Setup and Preliminaries}

We denote by $x(p)$ and $y(p)$ the $x$ and $y$ coordinates of a point $p$ in the plane.
We denote by $r(d)$ the radius of a disk $d$ in the plane.

Let $d^*$ be the smallest disk (not necessarily in $D$) that intersects every disk in $D$.
We find a set of four points that stab $D\cup \{d^*\}$.
Thus, without loss of generality, we assume that $d^*$ is in $D$.
By scaling and translating the scene, we may assume that $r(d^*)=1$ and that $c^*$, $d^*$'s center, is at the origin $(0,0)$.
Let $D^{-}$ be the set of all disks in $D$ that do not contain $c^*$.
For a disk $d \in D^{-}$, let $\delta(d) = |c c^*|- r(d)$, where $c$ is the center of $d$; see Figure~\ref{fig:alpha} for an illustration.
We denote by $D^{-}_{ \leq k}$, the set of all disks in $D^{-}$ whose radius is at most $k$.

Our choice of $d^*$ ensures that it is tangent to three disks, say $d_1,d_2,d_3$, of $D$.
Let $x_1,x_2,x_3$ be the points of tangency of $d^*$ with $d_1,d_2,d_3$, respectively.
Our choice of $d^*$ also ensures that $c^*$ is in the triangle $\triangle x_1 x_2 x_3$.
For each $i \in \{1,2,3\}$, we denote by $\ell_i$ the tangent line to $d^*$
at $x_i$, by $\ell_i'$ the reflection of $\ell_i$ with respect to
$c^*$, and by $h_i$ the closed halfplane defined by $\ell_i$ that contains $d_i$.
See Figure~\ref{fig:notation}.

\begin{figure}[hbt]
	\centering
	\includegraphics[scale=0.55]{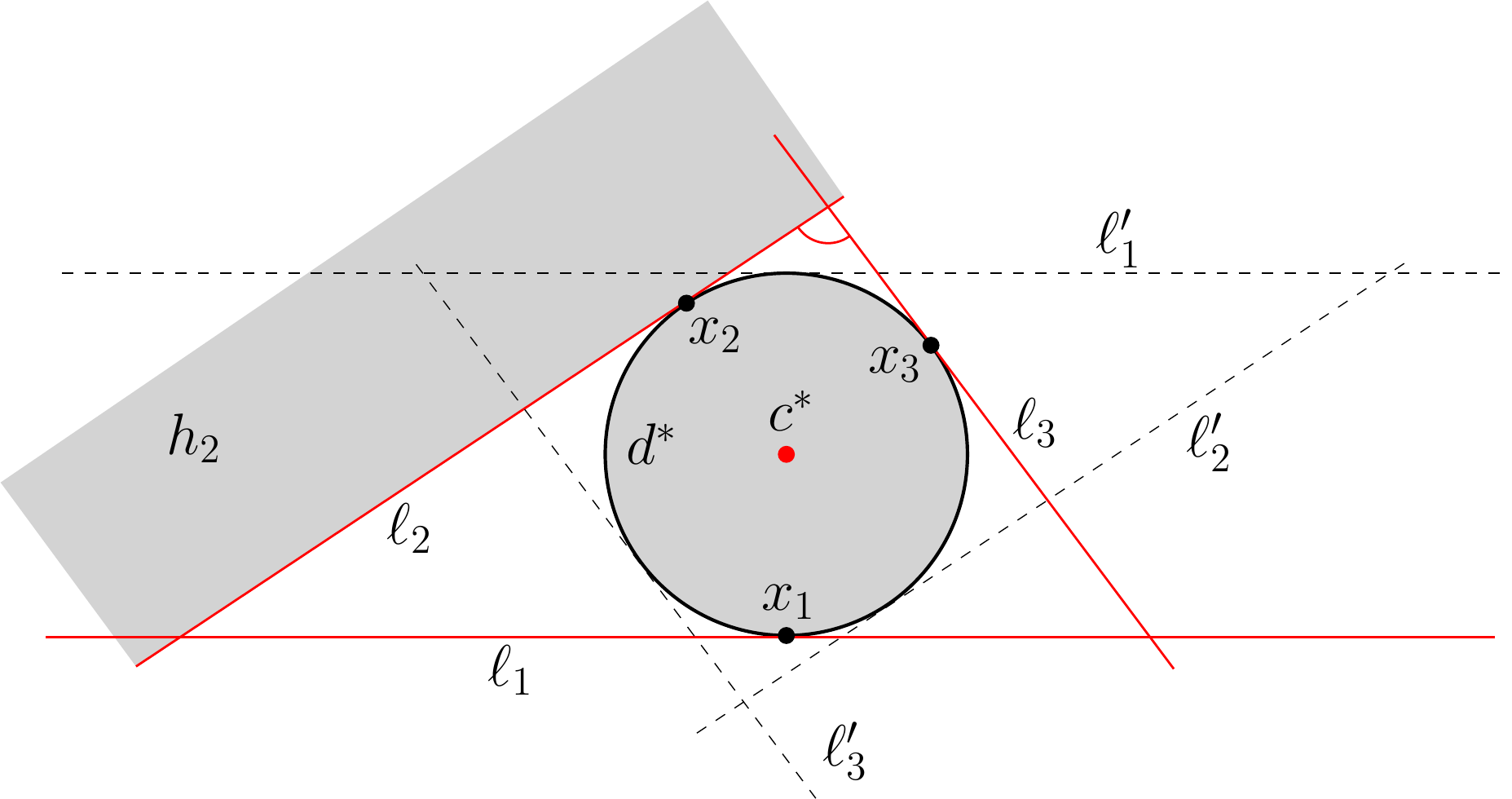}
	\caption{Setup and notation.}
	\label{fig:notation}
\end{figure}

{\color{mycolor}Base setting:}
Let $\Delta$ be the triangle formed by the intersection points of $\ell_1$, $\ell_2$, and $\ell_3$.
After rotating the scene and relabeling the disks $d_1,d_2,d_3$, we may assume that $x_1=(0,-1)$, the line $\ell_2$ has positive slope, the line $\ell_3$ has negative slope, and the largest angle of $\Delta$ is at the intersection of $\ell_2$ and $\ell_3$, as depicted in Figure~\ref{fig:notation}.


	{\color{mycolor}Alternative setting:}  Here we assume without loss of generality that the center of
	$d'$ is on the positive $x$-axis, where $d'$ is a disk in $ D^{-}_{ \leq k}$ with maximum $\delta(d')$, for some value of $k$.
	Later, whenever we use this setting, we specify the appropriate value of $k$.


\section{Stabbing by Five Points}
In this section we present a very simple linear-time algorithm for stabbing pairwise intersecting disks with five points. Its proof is similar in spirit to the proof of the much more involved stabbing by four points algorithm, thus it also serves as an introduction to the proof of the latter algorithm.
In this section we assume the {\color{mycolor}base setting}.

Let $S = \{(0,0), (2,0), (-2,0), (0,2), (0,-2)\}$. Then, $S$ is the set of corners plus center of the square of edge length $2\sqrt{2}$, whose center is at the origin and whose corners lie on the axes.
We claim that $S$ stabs $D$.

Indeed, let $d \in D$ be a disk and let $c$ denote its center. We prove that $d$ is stabbed by one of the points in $S$.
Assume first that $c$ lies in the second quadrant (i.e., in the north-west quadrant).
We distinguish between two cases.

\noindent
Case~1. If $-2 \le x(c) \le 0$, then since $d$ intersects $d_1$ it must also intersect $\ell_1$, and by Observation~\ref{radius}, $d$ is stabbed by $(0,0)$ or $(-2,0)$. To see this, let $(-2,0)$, $(0,0)$, and the $x$-axis play the roles of $a$, $b$, and $\ell$ in Observation~\ref{radius}, respectively. Moreover, let the closed halfplane below $\ell_1$ and the unit disk centered at $(-1,0)$ play the roles of $\eta$ and $\delta$, respectively, and let $d$ (whose radius is clearly at least 1) play the role of $\epsilon$.

\noindent
Case~2. If $x(c) < -2$, then $c$ lies above the line $\ell$ that passes through the points $(-2,0)$ and $(0,2)$. Let $V$ be the strip that is defined by the lines perpendicular to $\ell$ through $(-2,0)$ and through $(0,2)$, respectively.
If $c$ is outside $V$, then by Observation~\ref{outside-strip}, $d$ is stabbed by $(-2,0)$ or $(0,2)$. To see this, let $(-2,0)$, $(0,2)$, $\ell$, and $d^*$ and $d$ play the roles of $a$, $b$, $\ell$, and $\eta$ and $\epsilon$ in Observation~\ref{outside-strip}, respectively.
Hence, assume that $c$ is in $V$ (and above $\ell$).

We show that $d$ is stabbed by one of the points $(-2,0)$ and $(0,2)$.
Assume to the contrary that $d$ avoids both $(-2,0)$ and $(0,2)$.
Let $d'$ be the disk tangent to $d^*$ and whose boundary passes through $(-2,0)$ and $(0,2)$. Let $c'=(-x,x)$ be the center of $d'$. Then, we have
\begin{enumerate}[leftmargin=1.5\parindent]
		\item $-x < -2$, $x \ge 0$;
		\item $(-x - 0)^2 + (x - 2)^2 = r(d')^2$;
		\item $(-x - 0)^2 + (x - 0)^2 = (r(d') + 1)^2$.
\end{enumerate}
Solving these equations, we get that $x = \frac{3}{2} + \frac{3}{2\sqrt{2}}$ and $r(d') = \frac{1}{2} + \frac{3}{\sqrt{2}}$.
We show below that $r(d) \ge r(d')$ and apply Observation~\ref{radius} to conclude that $d$ is necessarily stabbed by $(-2,0)$ or $(0,2)$ --- contradiction. To apply Observation~\ref{radius}, we let $(-2,0)$, $(0,2)$, $\ell$, $d^*$ and $d'$ play the roles of $a$, $b$, $\ell$, $\eta$ and $\delta$, respectively.

It remains to show that $r(d) \ge r(d')$. Consider the disk $d''$, the smallest disk with center in $V$ (to the left of the line $x=-2$) that avoids both $(-2,0)$ and $(0,2)$ and is tangent to $\ell_1$. Clearly, $r(d) \ge r(d'')$, but as the calculation below shows $r(d'') = 2+\sqrt{2} \ge r(d')$, so $r(d) \ge r(d')$.

Indeed, let $c'' = (a, b)$ be the center of $d''$. Observe that the boundary of $d''$ must pass through $(-2,,0)$, since, otherwise, we can shrink $d''$, by moving $c''$ slightly to the right and down. Moreover, $c''$ must lie on $V$'s boundary, i.e., on the line $y = -x -2$, since otherwise we can shrink $d''$ by moving $c''$ slightly to the left and down. Thus, we have
\begin{enumerate}[leftmargin=1.5\parindent]
		\item $-a < -2$, $b \ge 0$;
		\item $b = -a - 2$;
		\item $(a + 2)^2 + (b - 0)^2 = r(d'')^2$;
		\item $b + 1 = r(d'')$.
\end{enumerate}
Solving these equations, we get that $a = -3 - \sqrt{2}$, $b = 1 + \sqrt{2}$, and $r(d'') = 2 + \sqrt{2}$.
This completes the proof under the assumption that $c$ is in the second quadrant.

Assume now that $c$ lies in the third quadrant.
Let $p=(1,0)$ and $q=(0,1)$ and let $\ell_p$ be the vertical line through $p$ and let $\ell_q = \ell_1'$ be the horizontal line through $q$.
If $d$ is intersected by $\ell_q$ or by $\ell_p$, then, by either reflecting the scene with respect to the $x$-axis or by rotating it clockwise about the origin, we get the previous case where $c$ lies in the second quadrant.
We thus assume that neither $\ell_q$ nor $\ell_p$ intersect $d$.
We show below that $d$ is stabbed by $(0,0)$.

Since $d$ intersects the disk $d_3$, it intersects $\ell_3$. Consider the line through $c$ that is perpendicular to $\ell_3$, and let $y$ be its intersection point with $\ell_3$. Then $r(d) \ge |cy|$. Now, since both $\ell_q$ and $\ell_p$ do not intersect $d$, we have that the segment $\overline{cy}$ intersects $\overline{c^*p}$ or $\overline{c^*q}$. Assume, w.l.o.g., that $\overline{cy}$ intersects $\overline{c^*p}$ and denote the intersection point by $p'$. Clearly, $|cp'| \ge |cc^*|$ (as $\overline{cp'}$ is the largest edge in the triangle $\Delta cc^*p'$), and therefore $r(d) > |cp'| \ge |cc^*|$, so $d$ is stabbed by $c^*=(0,0)$.

Finally, the cases where $c$ lies in the first or in the fourth quadrant are analogous to the two cases considered above.

The analysis of the algorithm's running time is similar to that of the stabbing by four points algorithm, which can be found in Section~\ref{time-section}.

\section{Stabbing by Four points --- The Algorithm}

Find the disk $d^*$, add it to $D$, and compute the set $D^{-}$. Let $\dmin$ denote the smallest disk in $D^{-}$, and let $\rmin$ denote its radius.
We show how to compute a set $S$ of four points that stabs $D$.

\subsection{Outline and Intuition}
The algorithm distinguishes between three main cases, depending on the value of $\rmin$. In the first two cases we use the {\color{mycolor}base setting}, while in the third case we use both settings.

In the first case all the disks in $D^{-}$ are large, specially their radius is at least 4. (Recall that the radius of $d^*$ is 1). This is the easiest case, since it is relatively easy to show that the constraint that every disk in $D^{-}$ must intersect each of the disks $d_1, d_2, d_3$, implies that every disk in $D^{-}$ is stabbed by one of the points $(-4,1), (4,1), (0,-3)$.

In the second case there are small disks in $D^{-}$, i.e., disks with radius smaller than 2. Among the small disks, let $d$ be the disk whose boundary is the farthest from $c^*$ and let $c$ be its center. Then, we know that all the disks in $D^{-}$ intersect the disk of radius $\delta(d)$ and center $c^*$. We use this fact, together with information on the (acute) angle $\alpha$ between the segment $cc^*$ and the $x$-axis to determine the set $S$. Specifically, assume without loss of generality that $x(c) \ge 0$. If $\alpha$ is small, then $c$ (which is either above or below the $x$-axis) is relatively close to the $x$-axis, and we find a set of four points that stabs $D$. If, on the other hand, $\alpha$ is large, then we have two symmetric sub-cases, depending on whether $c$ is above or below the $x$-axis.

In the third case there are no small disks in $D^{-}$, but there are disks of medium size, i.e., of radius between 2 and 4. This is the most difficult case, and we split it into four sub-cases. If there exists a disk $d \in D^{-}$ whose boundary is relatively far from $c^*$, then, if there exists such a disk which is not too large, then we are in sub-case 1, and otherwise we are in sub-case 2, provided there exists such a disk which is not huge. If we did not land in one of theses two sub-cases, we proceed to check whether there exists a not too large disk $d \in D^{-}$ whose boundary is close but not extremely close to $c^*$. If there exists such a disk, then we are in sub-case 3 and otherwise we are in sub-case 4. Our partition into sub-cases is not unique, in the sense that some of the values that we specify to distinguish between the cases are somewhat arbitrary, but we strongly believe that some similar partition into sub-cases is necessary.

The proof of correctness is essentially based on a small set of elementary geometric observations (see Section~\ref{sec:geom_observations}), and it is similar in spirit to the proof of the stabbing by five points algorithm. The ingenuity of the proof is reflected in the creative ways in which the settings to which these observations can be applied are defined.

\subsection{The Algorithm}

We consider three cases, depending on the value of $\rmin$.


\begin{itemize}
	\item $\rmin \geq 4$: Assuming the {\color{mycolor}base setting}, set $S=\{(0,0), (-4,1), (4,1), (0,-3)\}$.

	\item $\rmin \le 2$: Assuming the {\color{mycolor}base setting}, let $d \in D^{-}_{ \leq 2}$ be the disk with maximum $\delta(d)$ and let $c$ be its center. We may assume that $x(c) \ge 0$, since this can be guaranteed by a suitable reflection and relabeling.
	Let $\alpha$ be the convex angle between the segment $cc^*$ and the $x$-axis; see Figure~\ref{fig:alpha}.
	Depending on $\alpha$ and the $y$-coordinate of $c$, do

	\begin{enumerate}[leftmargin=1.5\parindent]
		\item $\alpha\le 17^{\circ}$: Set $S= \{ (-0.5,0), (0, -1.7),(0,1.7),(1.5,0) \}$.
		\item $\alpha > 17^{\circ}$ and $y(c) > 0$: Set $S= \left\{\left(-0.5,0\right),(0.5, -2.5),(-0.5,1.83),\left(\frac{1}{2} + \frac{2 \sqrt{6} }{5} ,\frac{1}{5}\right) \right\}$.
		\item  $\alpha > 17^{\circ}$ and $y(c) < 0$: Set $S = \left\{ (-0.5,0),(0.5, 2.5),(-0.5,-1.83),\left(\frac{1}{2} + \frac{2 \sqrt{6} }{5} ,-\frac{1}{5}\right) \right\}$.
	\end{enumerate}

	\item $2 < \rmin < 4$: This case involves four sub-cases. In the first three sub-cases,
	we assume the {\color{mycolor} alternative setting}, while in the fourth we assume the {\color{mycolor}base setting}. Thus, in the first three sub-cases we pick a specific disk $d'$ and assume that its center is on the positive $x$-axis. (This can be guaranteed by a suitable rotation.)

	\begin{enumerate}[leftmargin=1.5\parindent]
		\item If there exists a disk $d \in D^{-}_{ \leq 5}$ with $\delta(d) \geq 0.5$, then let $d'$ be such a disk with maximum $\delta(d')$ and set $S=\{(0,0), (2,0), (0.4,2), (0.4,-2)\}$.
		\item Else, if there exists a disk $d \in D^{-}_{ \leq 20}$ with $\delta(d) \geq 0.5$, then let $d'$ be such a disk with maximum $\delta(d')$ and set $S=\{(0,0), (2,0), (-0.15,2.7), (-0.15,-2.7)\}$.

		\item Else, if there exists a disk $d \in D^{-}_{ \leq 5}$ with $\delta(d) \geq 0.11$, then let $d'$ be such a disk with maximum $\delta(d')$ and set $S=\{(0,0), (2,0), (-0.15,1.75), (-0.15,-1.75)\}$.

		\item Otherwise, set $S=\{(0,0), (2.5,1), (-2.5,1), (0,-1.52)\}$, assuming the {\color{mycolor}base setting}.
	\end{enumerate}

\end{itemize}

This is the end of our algorithm and the computation of $S$. We claim that $S$ stabs $D$. In Section~\ref{time-section} we establish the linear bound on the algorithm's running time, and in Section~\ref{correctness-section} we prove its correctness.

\begin{figure}[hbt]
	\centering
	\includegraphics[scale=0.8]{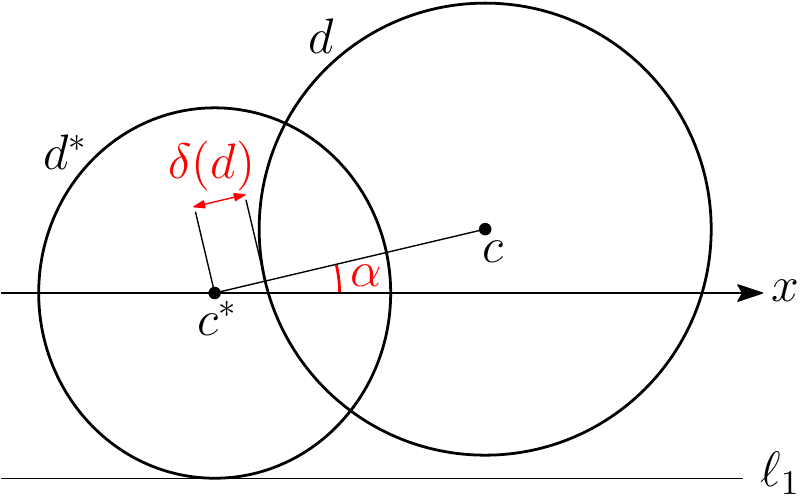}
	\caption{Illustration of $\delta(d)$ and the angle $\alpha$.}
	\label{fig:alpha}
\end{figure}

\subsection{Running Time}
\label{time-section}
All transformations, i.e., translations, rotations, and reflections, can be performed in $O(|D|)$ time.
The smallest disk $d^*$ that intersects all disks in $D$, can be found in $O(|D|)$ time, as it is an LP-type problem; see~\cite{CM96,MSW96}.
The smallest disk $\dmin$ that does not contain $c^*$, as well as the disks $d$ and $d'$ with maximum $\delta(d)$ and $\delta(d')$, can be found in $O(|D|)$ time. Given $d^*$, $\dmin$, $d$, and $d'$, the set $S$ can be found in constant time. Therefore, the total running time of the algorithm is $O(|D|)$.

\section{Correctness}
\label{correctness-section}
In this section we prove the correctness of the algorithm, that is, we prove that the set $S$ computed by the algorithm stabs $D$.
We prove this for the three cases ($\rmin \ge 4$, $\rmin \leq 2$, and $2 < \rmin < 4$) separately in Subsections~\ref{secDminBig}, \ref{secDminLessThanTwo}, and \ref{secDminTwoToFour}. Recall the three disks, $d_1$, $d_2$, and $d_3$, that are tangent to $d^*$. Any disk $d\in D$ that we consider in our proofs should intersect these three disks. Finally, the elementary geometric observations to which we refer in our proofs can be found in Section~\ref{sec:geom_observations}.

%
%
\subsection{Proof of Case $\rmin\ge 4$}
\label{secDminBig}
Recall that we are in the {\color{mycolor}base setting}, and $S=\{c^* =(0,0), (-4,1), (4,1), (0,-3)\}$ as depicted in Figure~\ref{fig:Case1}. We prove that each disk in $D$ contains at least one point of $S$.

\begin{figure}[hbt]

	\centering
	\includegraphics[scale=0.5]{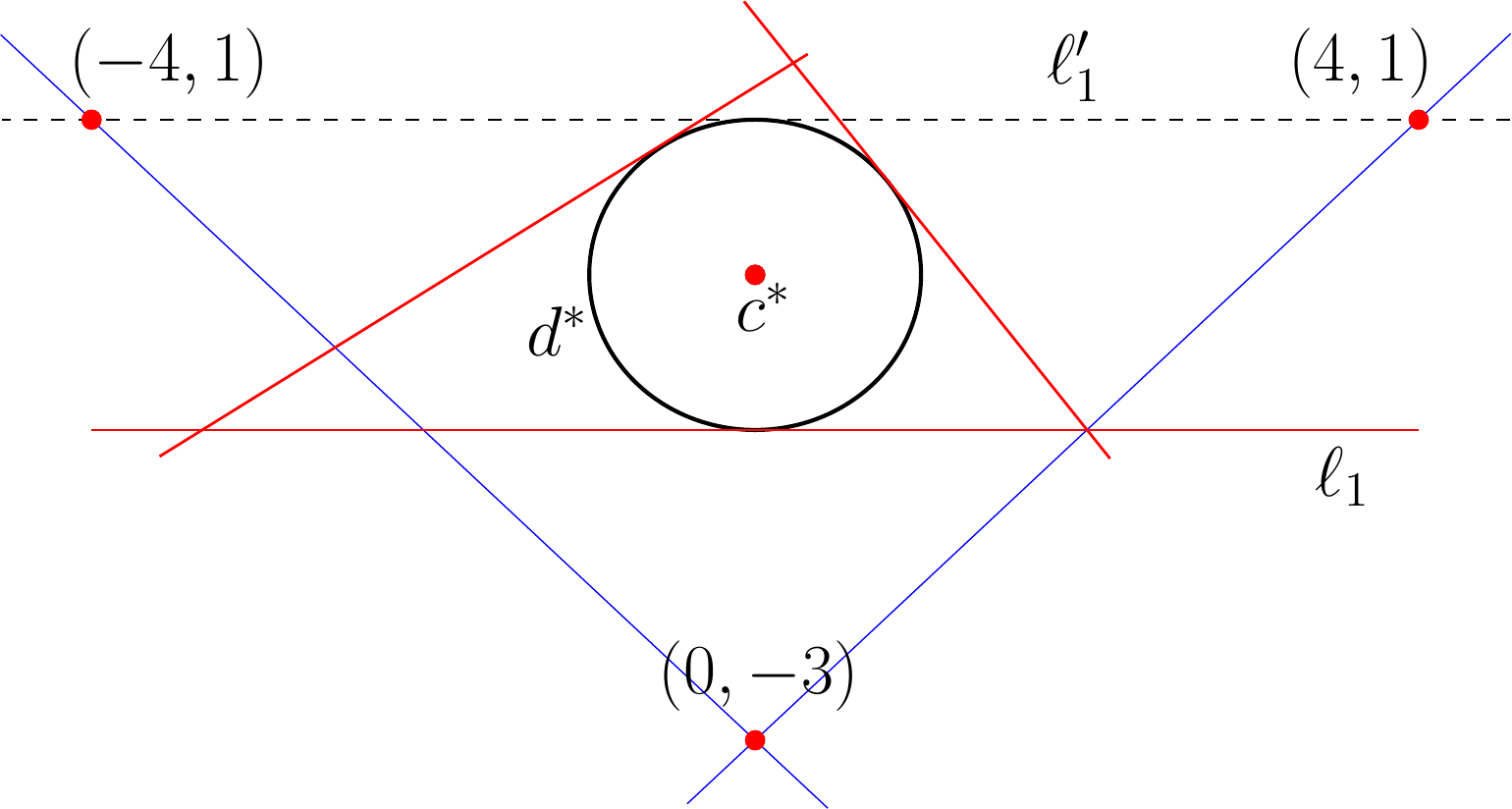}
	\caption{Four stabbing points for the case where $\rmin \geq 4$.}
	\label{fig:Case1}
\end{figure}

By the definition of $D^-$, any disk in $D\setminus D^-$ contains $c^*=(0,0)$. Now consider any disk $d\in D^-$ and let $c$ denote its center.
Since $\rmin \ge 4$, we have $r(d)\ge 4$.


\begin{enumerate}
	\item If $c$ lies below or on the line through $(4,1)$ and $(0,-3)$, then by Observations~\ref{outside-strip} and \ref{radius}, $d$ contains $(4,1)$ or $(0,-3)$.
	To see this, let $l$ be the line through $a=(4,1)$ and $b=(0,-3)$, and $V$ be the strip defined by the lines perpendicular to $l$ through $a$ and $b$, respectively.
	The role of the disk $\eta$ in these observations is filled by the disk $d^*$, which is contained in the lower part of $V$.
	If $c$ (the center of $d$) lies outside $V$, then, by Observation~\ref{outside-strip}, $d$ contains $(4,1)$ or $(0,-3)$.
	Otherwise ($c$ lies inside $V$), consider the disk $f$ of radius 4 centered at $(4,-3)$ (which is in the upper part of $V$), and notice that its boundary passes through $a$ and $b$ and is tangent to $d^*$. Now, by Observation~\ref{radius}, where the role of the disk $\delta$ is filled by the disk $f$, $d$ contains $(4,1)$ or $(0,-3)$,
	since $r(d) \ge r(f) = 4$.
	\item If $c$ lies below or on the line through $(-4,1)$ and $(0,-3)$, then by an argument similar to the previous case, $d$ contains $(-4,1)$ or $(0,-3)$.
	\item If $c$ lies on or above the line through $(-4,1)$ and $(4,1)$ (i.e., the line $\ell'_1$) and $x(c) \notin [-4,4]$, then, by Observation~\ref{outside-strip}, $d$ contains $(-4, 1)$ or $(4,1)$, where the role of the disk $\eta$ is filled by the disk $d^*$.
	\item If $c$ lies on or above $\ell'_1$ and $x(c) \in [-4,4]$, then we distinguish between two simple cases.
	If $y(c) \le 3$, then it is easy to verify that $d$ contains at least one of the points $(-4,1)$ and $(4,1)$ (since $c$ lies in the union of the disks of radius 4 around these points, as $d \in D^{-}$).
	If $y(c) > 3$, then, by Observation~\ref{b}, $d$ contains $(-4,1)$ or $(4,1)$. To see this, assume w.l.o.g. that $x(c) \in [0,4]$, and let $c^*=(0,0)$, $(4,1)$, and $\ell_1$ play the roles of $a$, $b$, and $l$ in Observation~\ref{b}, respectively. Then $V$ is the vertical strip whose sides pass through $a$ and $b$, respectively. Moreover, let $d'$ be the disk with center in $V$ whose boundary passes through $c^*$ and $(4,1)$ and is tangent to $\ell_1$ (see Figure~\ref{fig:biggerThanFour4}); it is easy to verify that $d'$'s center lies below the horizontal line $y=2$. Then, $d'$ and $d$ play the roles of $\delta$ and $\epsilon$ in Observation~\ref{b}, and according to the observation $d$ contains $c^*$ or $(4,1)$. But since we are assuming that $c^* \notin d$, we conclude that $d$ contains $(4,1)$.
	\item  If $c$ is in the triangle with corners $(-4,1)$, $(4,1)$, and $(0,-3)$, then since $\rmin \ge 4$ it is stabbed by at least one of the points in $S$.
\end{enumerate}

\begin{figure}[hbt]
	\centering
	\includegraphics[scale=0.4]{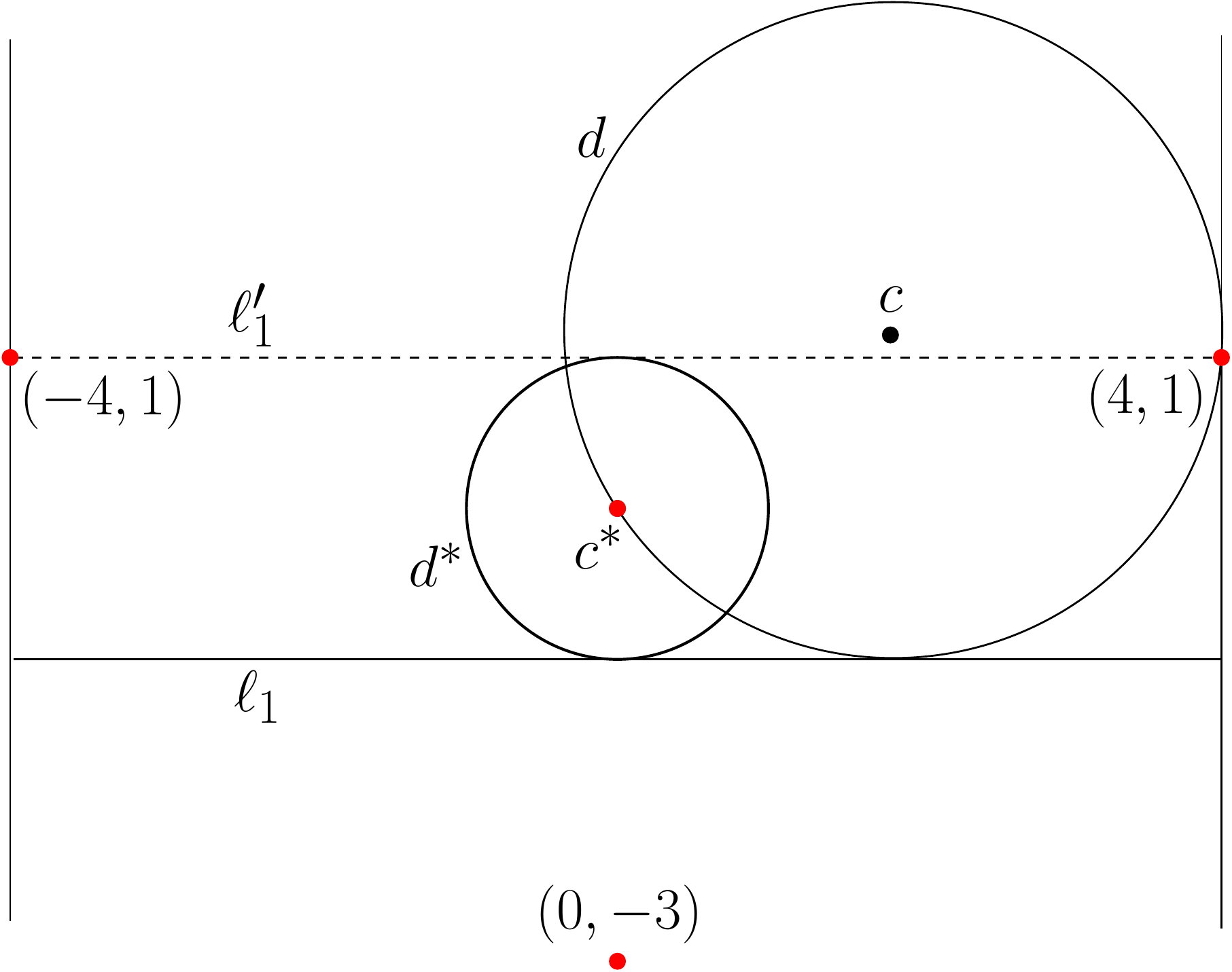}
	\caption{Illustration of the case where $c$ lies above $\ell'_1$ and $c(x) \in [-4,4]$ (more precisely, $x(c) \in [0.4]$ and $y(c) > 3$).}
\old{
	\caption{Illustration of the case where $c$ lies on or above $\ell'_1$ and $x(c)\in [0,4]$.}
}
	\label{fig:biggerThanFour4}
\end{figure}

%
\subsection{Proof of Case $\rmin\le 2$}
\label{secDminLessThanTwo}
Recall that we are in the {\color{mycolor}base setting}, that
$d$ is the disk in $D^{-}_{ \leq 2}$ with maximum $\delta(d)$, that
$d$'s center, $c$, has positive $x$-coordinate, and that $\alpha$ is the convex angle between the segment $cc^*$ and the $x$-axis (see
Figure~\ref{fig:alpha}).
Depending on $\alpha$ and $y(c)$, our algorithm sets $S=\{p^l,q^-,q^+,p^r\}$ as follows:

\begin{enumerate}[leftmargin=1.5\parindent]
	\item $\alpha\le 17^{\circ}$: \[p^l=(-0.5,0), \ q^-=(0, -1.7), \ q^+=(0,1.7),\ p^r=(1.5,0).\]
	\item $\alpha > 17^{\circ}$ and $y(c) > 0$:
	\[p^l=(-0.5,0), \ q^-=(0.5, -2.5), \ q^+=(-0.5,1.83), \ p^r=\left(\frac{1}{2} + \frac{2 \sqrt{6} }{5} ,\frac{1}{5}\right).\]
	\item  $\alpha > 17^{\circ}$ and $y(c)< 0$:
	\[p^l=(-0.5,0), \ q^+=(0.5, 2.5), \    q^-=(-0.5,-1.83), \ p^r=\left(\frac{1}{2} + \frac{2 \sqrt{6} }{5} ,-\frac{1}{5}\right).\]
\end{enumerate}

We prove for each of the above cases that $S$ stabs $D$.

Any disk $e \in D$ intersects the line $\ell_1$.
To see this, if $e$  is centered at a point in $h_1$,
where $h_1$ is the half plane defined by $\ell_1$ that does not contain $c^*$, then since $e$ must intersect $d^*$ it intersects $\ell_1$.
If $e$  is centered at a point not in $h_1$, then since $e$ must intersect $d_1$ it intersects $\ell_1$.
Moreover, if $e$'s center is with positive $y$-coordinate, then it also intersects $\ell'_1$.
In Section~\ref{sec:blabla} we prove the following claim.
\begin{claim}\label{cl:blabla}
Any disk in $D$ that is centered at a point of negative $y$-coordinate and is not stabbed by $q^-$ nor $p^l$
intersects $\ell'_1$.
\end{claim}
Therefore, we may restrict our attention only to disks in $D$ that intersect both $\ell_1$ and $\ell'_1$.
Consider the vertical strip $V$ whose left and right boundaries are the vertical lines through $p^l$ and $p^r$, respectively.

\begin{claim}\label{cl:inStripV}
Any disk $e\in D$, with center in $V$, contains $p^l$ or $p^r$.
\end{claim}
\begin{proof}
We distinguish between two cases w.r.t. $\alpha$ and apply Observation~\ref{b} in these two cases.
In both cases, if the center of $e$ is above the $x$-axis then $\ell_1$ plays the role of $\ell$ (in Observation~\ref{b}),
otherwise $\ell'_1$ plays the role of $\ell$.

If $\alpha \leq 17^{\circ}$, then $x(p^r) =\frac{3}{2}$ and the claim follows from Observation~\ref{b} where
$d^*$,      $e$,       $p^l$, and $p^r$ play the roles of
$\delta$, $\epsilon$, $a$, and   $b$, respectively.

If $\alpha > 17^{\circ}$ then $x(p^r)= \frac{1}{2}+ \frac{2 \sqrt{6}}{5}$.
Let $d'$ be the disk (not necessarily in $D$) of radius $1$ that is centered at $(0.5,0)$.
Observe that the points $\left(\frac{1}{2} + \frac{2 \sqrt{6} }{5} ,\frac{1}{5}\right)$ and
$\left(\frac{1}{2} + \frac{2 \sqrt{6} }{5} , - \frac{1}{5}\right)$ lie on the boundary of $d'$.
The claim follows from Observation~\ref{b} where  $d'$,     $e$,        $p^l$, and $p^r$ play the roles of
$\delta$, $\epsilon$, $a$, and   $b$, respectively.
\end{proof}

Let $Q_1$, $Q_2$, $Q_3$, and $Q_4$ be the four quadrants of the plane in counterclockwise order around $c^*$
such that $Q_1$ contains all points with positive $x$- and $y$-coordinates.
We define four disks $d^{r+}$, $d^{l+}$, $d^{l-}$, and $d^{r-}$ as follows
(see Figures~\ref{fig:SmallerThan2Case1} and \ref{fig:SmallerThan2Case2}):

\begin{itemize}
	\item $d^{r+}$ is the disk, with center in $Q_1$,
	with $x$- and $y$-coordinates more than 1,
	that is tangent to $\ell_1$, and has $p^r$ and $q^+$ on its boundary.

	\item $d^{l+}$ is the disk, with center in $Q_2$, that is tangent to  $\ell_1$, and has $p^l$ and $q^+$ on its boundary.

	\item $d^{l-}$ is the disk, with center in $Q_3$, that is tangent to  $\ell'_1$, and has $p^l$ and $q^-$ on its boundary.

	\item $d^{r-}$ is the disk, with center in $Q_4$
	with $x$-coordinate more than 1 and with $y$-coordinate less than $-1$, that is tangent to  $\ell'_1$, and has $p^r$ and $q^+$ on its boundary.
\end{itemize}

In Sections~\ref{alpha-small},~\ref{alpha-big-y-positive} and~\ref{alpha-big-y-negative} we prove the following two lemmas.

\begin{lemma}\label{lem:qLieAboveTangent}
 Point $q^+$ lies above the mutual tangent of $d^*$ and $d$.
\end{lemma}

\begin{lemma}\label{lem:intersect1}
  Disks $d^{l+}$ and $d^{l-}$ do not intersect $d$; and
  disks   $d^{r+}$ and $d^{r-}$ do not intersect $d^*$.
\end{lemma}

The following lemma together with Claim~\ref{cl:inStripV} implies that $S=\{p^l,p^r, q^+, q^- \}$ stabs $D$.
\begin{lemma}\label{lem:mainRlessThan2}
Any disk $e \in D$, whose center is outside $V$, contains a point of $S$.
\end{lemma}
\begin{proof}
We consider four cases according to which quadrant (with apex at $c^*$) the center of $e$ belongs to.

\noindent
\textbf{Case 1:} $e$'s center is in the second quadrant and outside $V$.
If $q^+$ lies above $\ell'_1$, then
we have by Corollary~\ref{smaller} that $e$ contains $p^l$ or $q^+$.
To see that,     $p^l$, $q^+$, $d$,      $e$,       $d^{l+}$ and $\ell_1$,
play the role of $a$,   $b$,   $\eta'$, $\epsilon$, $\delta$ and $\ell$
in Corollary~\ref{smaller}.
Else, $q^+$ lies below $\ell'_1$ and above the mutual tangent of $d^*$ and $d$ as shown in Lemma~\ref{lem:qLieAboveTangent}.
We have by Observation~\ref{Obs:5Plus} that $e$ contains $p^l$ or $q^+$.
   To see that, $p^l$, $q^+$, $d$,    $d^{l+}$, $e$         $d^*$    and $\ell_1$,
play the role of $a$,  $b$,   $\eta$, $\delta$, $\epsilon$, $\gamma$ and $\ell$
in Observation~\ref{Obs:5Plus}.

\noindent
\textbf{Case 2:} $e$'s center is in the third quadrant and outside $V$.
Then,  since $e$ intersects $\ell'_1$ this case is analogous to the previous case.

\noindent
\textbf{Case 3:} $e$'s center is in the first quadrant and outside $V$.
Then, we have by Corollary~\ref{smaller} (reflected with respect to the $y$-axis)
that $e$ contains $p^r$ or $q^+$.
To see that,     $p^r$, $q^+$, $d^*$,      $e$,        $d^{r+}$ and $\ell_1$,
play the role of $a$,   $b$,   $\eta'$,    $\epsilon$, $\delta$   and $\ell$
in Corollary~\ref{smaller}.

\noindent
\textbf{Case 4:} $e$'s center is in the fourth quadrant and outside $V$.
Then,  since $e$ intersects $\ell'_1$ this case is analogous to the previous case.
\end{proof}

In the rest of this section we provide the proofs for the claims that we stated without proof.
We distinguish between the three main cases with respect to the location of $c$ ($d$'s center).
That is, (i) $\alpha \leq 17^{\circ}$ in Section~\ref{alpha-small}; (ii) $\alpha > 17^{\circ}$ and $y(c) > 0$ in Section~\ref{alpha-big-y-positive};
 (iii) $\alpha > 17^{\circ}$ and $y(c) < 0$ in Section~\ref{alpha-big-y-negative}.

\begin{figure}[hbt]
	\centering
	\includegraphics[scale=0.5]{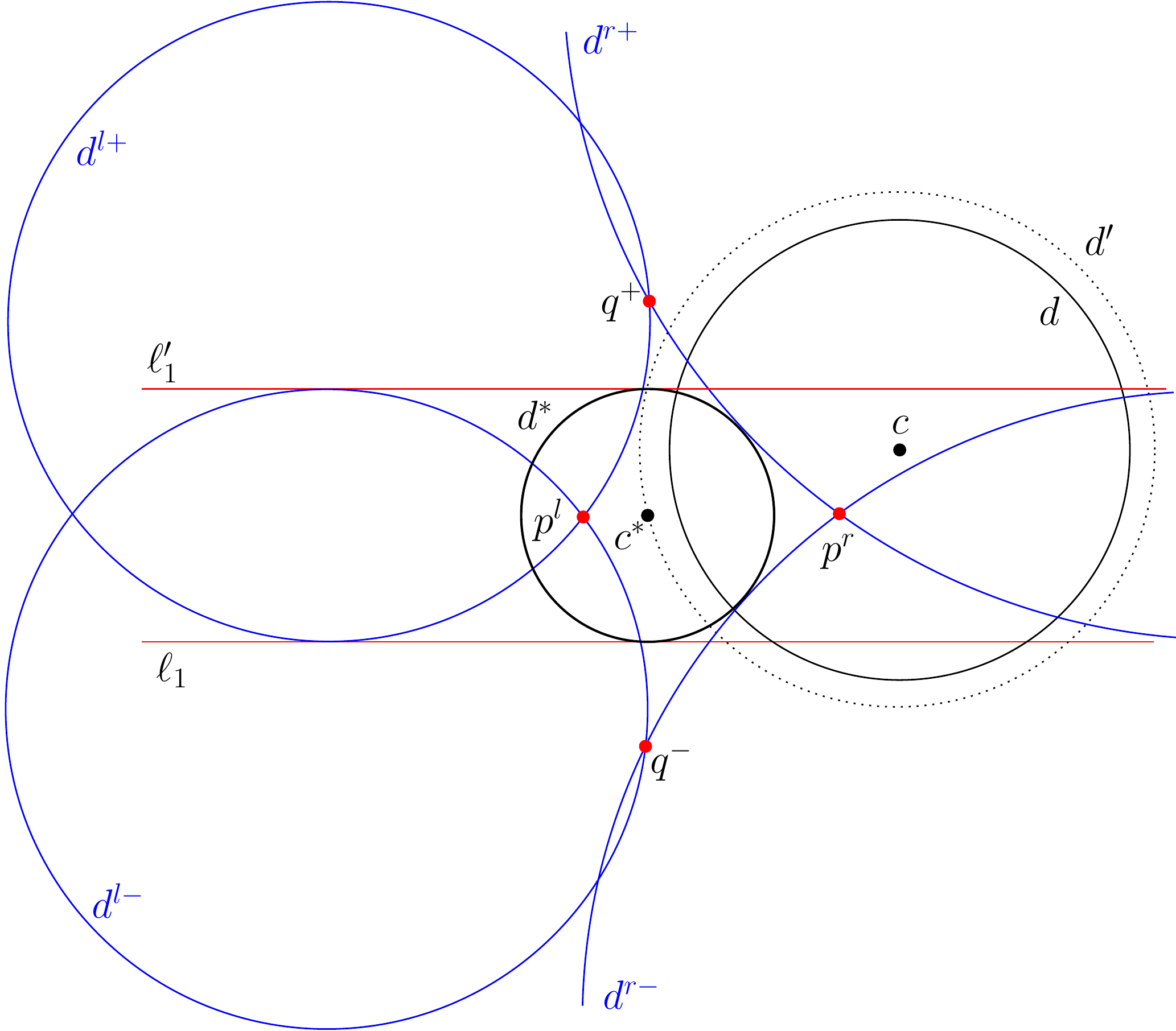}
	\caption{Illustration of $d$, $d^{r+}$, $d^{l+}$, $d^{l-}$, and $d^{r-}$ when $\alpha\le 17^{\circ}$.}
	\label{fig:SmallerThan2Case1}
\end{figure}

\subsubsection{$ \alpha \le 17^{\circ}$}
\label{alpha-small}

Recall that in this case $q^+ = (0,1.7)$. Let $t$ be the upper mutual tangent of $d^*$ and $d$. In order to be able to apply Observation~\ref{Obs:5Plus}, we first need to prove the following claim.

\begin{figure}[hbt]
	\centering
	\includegraphics[scale=0.75]{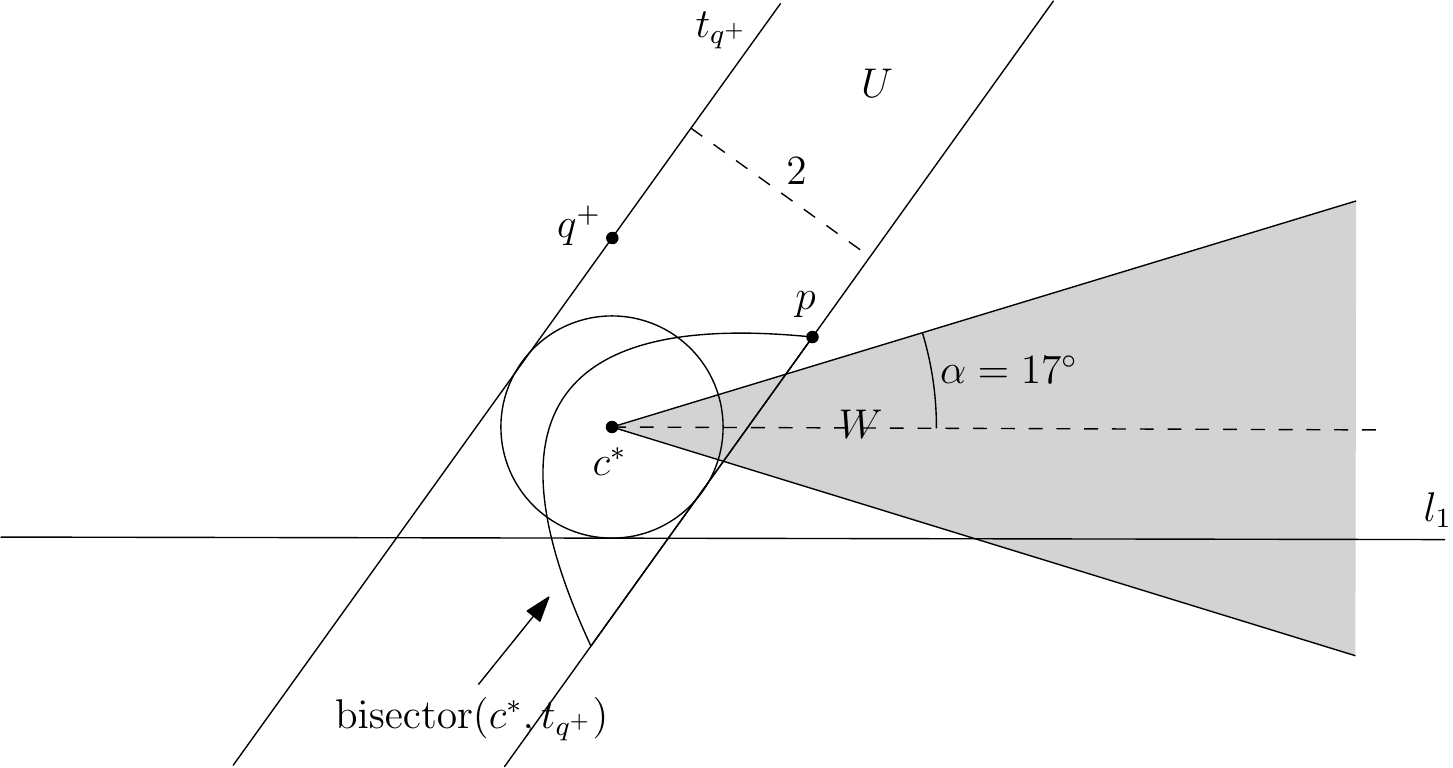}
	\caption{Proof of Claim~\ref{cl:q_plus_above1}.}
	\label{fig:q_plus_above1}
\end{figure}

\begin{claim}
\label{cl:q_plus_above1}
$q^+$ is above $t$.
\end{claim}
\begin{proof}
By our assumptions on the underlying setting we know that $x(c) > 0$. Moreover, $d$ intersects $d^*$ but is not stabbed by $c^*$, and $d$ intersects the line $l_1$, so $y(c) \le 1$ (since $r(d) \le 2$).
Let $t_{q^+}$ be the line with positive slope tangent to $d^*$ and passing through $q^+$; see Figure~\ref{fig:q_plus_above1}. If $d$ lies below $t_{q^+}$, then $q^+$ is clearly above $t$ and we are done. So assume that $d$ intersects $t_{q^+}$, i.e., the distance between $c$ and $t_{q^+}$ is at most $r(d)$. Since $c$ is below $t_{q^+}$ and $r(d) \le 2$, we conclude that $c$ lies in the strip $U$ defined by $t_{q^+}$ and the line parallel to $t_{q^+}$ at distance 2 from $t_{q^+}$ and below $t_{q^+}$. Moreover, $c$ is closer to $t_{q^+}$ than to $c^*$, so $c$ must lie on the `right' side of the bisector between $c^*$ and $t_{q^+}$. Finally, $c$ must also lie in the convex wedge $W$ defined by the rays emanating from $c^*$ of angles $17^\circ$ and $-17^\circ$, respectively. Our calculations below show that the intersection of these three regions is empty, leading to the conclusion that $d$ does not intersect $t_{q^+}$.

Indeed, $t_{q^+}$'s point of tangency is $(-\frac{3\sqrt{21}}{17},\frac{10}{17})$ and $t_{q^+}$'s equation is $y=\frac{6.3}{\sqrt{21}}x+1.7$. Thus, the equation of the lower boundary of the strip $U$ is $y=\frac{6.3}{\sqrt{21}}x-1.7$. The bisector between $c^*$ and $t_{q^+}$ is the parabola $P$:
\[ x^2 + y^2 = \frac{21}{6.3^2+21}\left(\frac{6.3}{\sqrt{21}}x - y + 1.7 \right)^2 \ . \]
Since $x(c) > 0$ and $c$ must lie in $U$ and on the `right' side of $P$, it is enough to show that any point satisfying these conditions does not lie in $W$. To see this, we calculate the point $p$ and the angle $\beta$ corresponding to it, where $p$ is the higher intersection point between $P$ and the lower boundary of $U$; see Figure~\ref{fig:q_plus_above1}. We get that
$p = (\sqrt{3}(\frac{10}{17}+\frac{3\sqrt{7}}{17}),\frac{9\sqrt{7}}{17}-\frac{10}{17})$ and $\beta >23^\circ$, and therefore any point satisfying these conditions corresponds to an angle greater than $\beta$.
\end{proof}

	First we prove that $d^{l+}$ and $d^{l-}$ do not intersect $d$. Because of symmetry, we prove this only for $d^{l+}$. To that end, let $(a,b)$ denote the center of $d^{l+}$; notice that $a\le 0$ and $b\ge 0$. Since $d^{l+}$ is tangent to $\ell_1$ and has $p^l=(-0.5,0)$ and $q^+=(0,1.7)$ on its boundary, the following equations hold:
	\begin{enumerate}[leftmargin=1.5\parindent]
		\item $r(d^{l+})=b+1$,
		\item $(a + 0.5 )^2 + b^2=  r(d^{l+})^2$,
		\item $a^2 + (b - 1.7)^2 =  r(d^{l+})^2$.
	\end{enumerate}

	These equations solve to:	$ a = \frac{-27}{34} - \frac{3 \sqrt{{471}/{5}} }{17}, b= \frac{2919}{2890} + \frac{3 \sqrt{2355 } }{289}.$

	Let $d'$ be the disk of radius 2 that has $c^*$ on its boundary and has its center $c'=(x',y')$ on the ray from $c^*$ to $c$, i.e., its center is on a line through the origin that makes angle $\alpha $ with the $x$-axis; see Figure~\ref{fig:transfrom4}. If fact we have $x'= 2\cos \alpha$ and $y'=2\sin \alpha$. We claim that if $d^{l+}$ does not intersect $d'$, then it also does not intersect $d$. To verify this claim we consider two cases: (i) $c$ is on the line segment $c^*c'$ (ii) $c'$ is on the line segment $c^*c$.
	In case (i) our claim holds because $d'$ contains $d$; see Figure~\ref{fig:transfrom4}(a).
	In case (ii) our claim holds because the center of $d^{l+}$ is closer to $c'$ than to $c$ (the center of $d^{l+}$ and the point $c'$ are on the same side of the the perpendicular bisector of  $c'c$) and $r(d') \geq r(d)$;
	 see Figure~\ref{fig:transfrom4}(b).

	\begin{figure}[htb]
		\centering
		\setlength{\tabcolsep}{0in}
		$\begin{tabular}{cc}
		\multicolumn{1}{m{.5\columnwidth}}{\centering\includegraphics[width=.3\columnwidth]{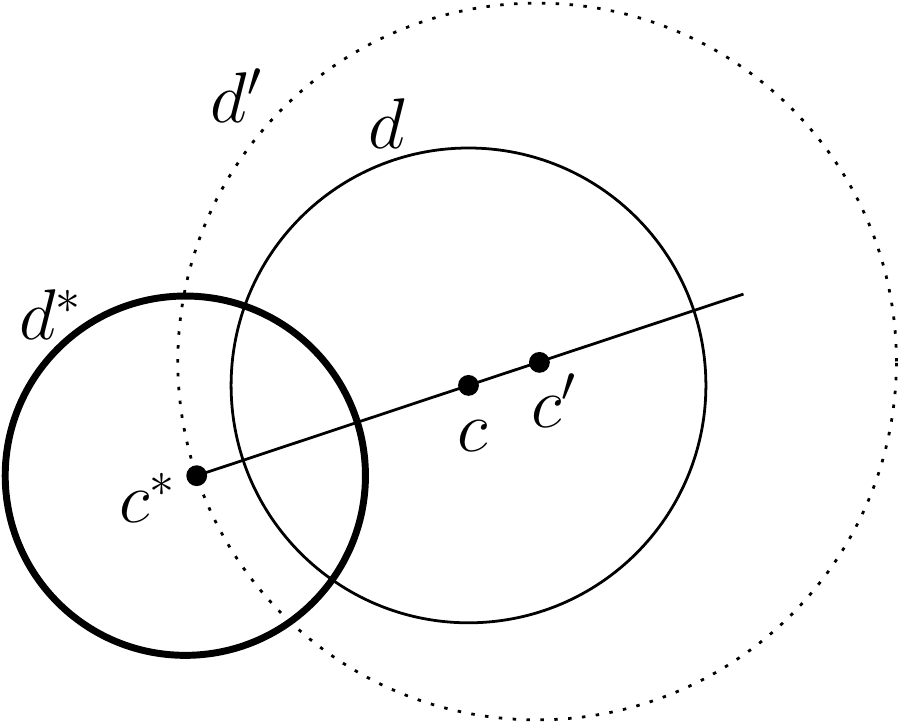}}
		&\multicolumn{1}{m{.5\columnwidth}}{\centering\includegraphics[width=.3\columnwidth]{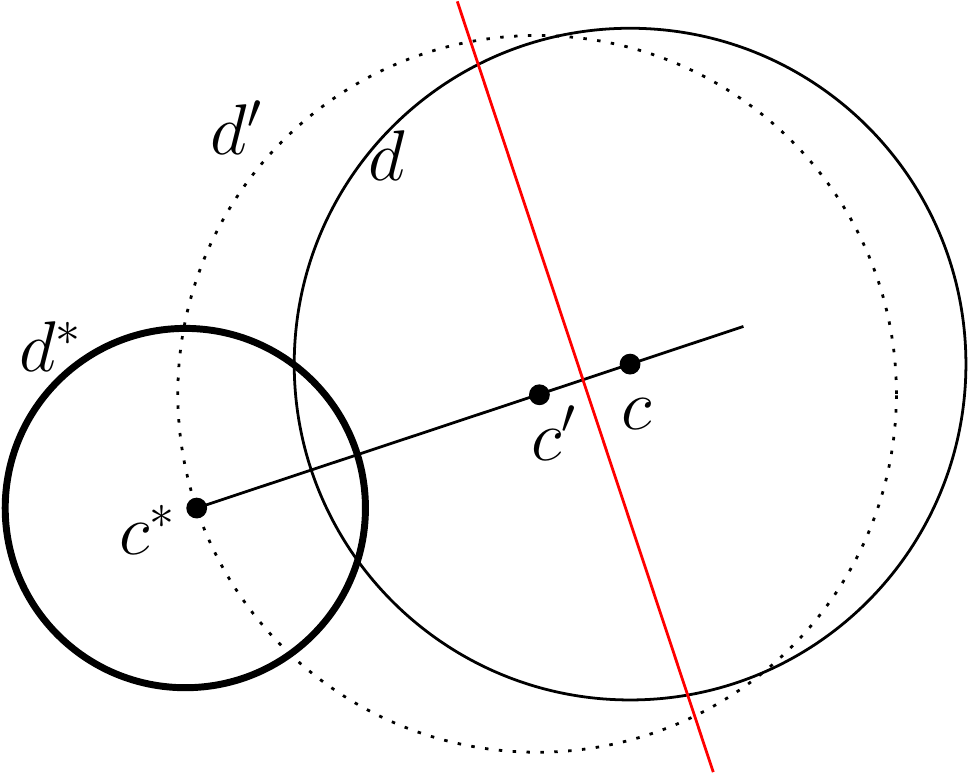}} \\
		(a) &(b)
		\end{tabular}$
		\caption{(a) The center $c$ lies on $c^*c'$, and (b) the center $c'$ lies on $c^*c$.}
		\label{fig:transfrom4}
	\end{figure}

%

	By the above claim it suffices to show that $d^{l+}$ does not intersect $d'$, that is the distance between their centers is more than the sum of their radii:
	$(a - x')^2 + (b -y')^2 > (2 + (b+1) )^2$. This can be verified by plugging $a$, $b$, $x'$, and $y'$ as follows:

	\begin{align}
	&\notag (a - x' )^2 + (b -  y')^2 = \ (a - x' )^2 + ( y' - b)^2 = \\
	&\notag \left(2 \cos \alpha +  \frac{27}{34} + \frac{3 \sqrt{{471}/{5}} }{17}   \right)^2 + \left(\frac{2919}{2890} + \frac{3 \sqrt{2355 } }{289} - 2 \sin \alpha \right)^2 \ge\\
	&\notag \left(2 \cos 17^{\circ} + \frac{27}{34} + \frac{3 \sqrt{{471}/{5}} }{17}   \right)^2 + \left(\frac{2919}{2890} + \frac{3 \sqrt{2355 } }{289} - 2 \sin 17^{\circ} \right)^2>\\
	&\notag \left(2 + \frac{2919}{2890} + \frac{3 \sqrt{2355 } }{289}  +1 \right)^2.
	\end{align}


	The above inequalities imply that $d^{l+}$ does not intersect $d'$. Therefore $d^{l+}$ does not intersect $d$.

	Now we prove that $d^{r+}$ and $d^{r-}$ do not intersect $d^*$. Because of symmetry, we show this only for $d^{r+}$.
	Let $(a,b)$ denote the center of $d^{r+}$, and recall that $a \ge 1 $ and $b \ge 1$. Since $d^{r+}$ is tangent to $\ell_1$ and has $p^r=(1.5, 0)$ and $q^+=(0,1.7)$ on its boundary, the following equations hold:
	\begin{enumerate}[leftmargin=1.5\parindent]
		\item $r(d^{r+})=b+1$,
		\item $(a - 1.5 )^2 + b^2=  r(d^{r+})^2$,
		\item $a^2 + (b - 1.7)^2 =  r(d^{r+})^2$.
	\end{enumerate}

	These equations solve to: $a = \frac{81}{34} + \frac{3 \sqrt{{771}/{5}} }{17}, b= \frac{6619}{2890} + \frac{9 \sqrt{3855 } }{289}.$

	To prove that $d^{r+}$ and $d^*$ do not intersect, it suffices to show that the distance between their centers is more than the sum of their radii: $ (a - 0)^2 + (b - 0)^2 > (1 + (b+1) )^2$. This can be verified
	by plugging $a$ and $b$ as follows:

	\[  \left(\frac{81}{34} + \frac{3 \sqrt{{771}/{5}} }{17}  \right)^2 + \left(\frac{6619}{2890} + \frac{9 \sqrt{3855 } }{289} \right)^2  > \left(1  + \frac{6619}{2890} + \frac{9 \sqrt{3855 } }{289}  +1 \right)^2.\]

	The above inequality implies that $d^{r+}$ does not intersect $d^*$.

\subsubsection{$ \alpha > 17^{\circ}$ and $y(c) > 0$}
\label{alpha-big-y-positive}

Recall that in this case $q^+ = (-0.5,1.83)$. Let $t$ be the upper mutual tangent of $d^*$ and $d$. In order to be able to apply Observation~\ref{Obs:5Plus}, we first need to prove the following claim.

\begin{figure}[hbt]
	\centering
	\includegraphics[scale=0.75]{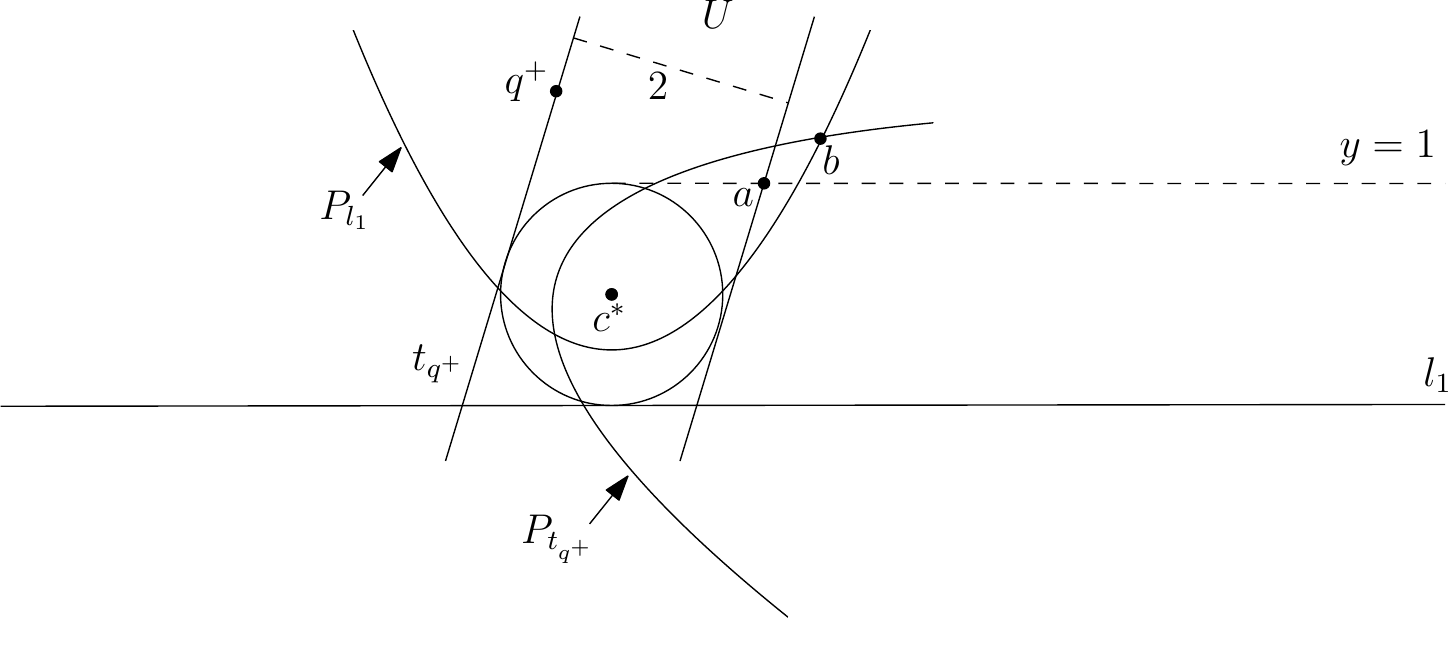}
	\caption{Proof of Claim~\ref{cl:q_plus_above2}.}
	\label{fig:q_plus_above2}
\end{figure}

\begin{claim}
\label{cl:q_plus_above2}
$q^+$ is above $t$.
\end{claim}
\begin{proof}
By our assumptions on the underlying setting we know that $c$ is in the first quadrant. Moreover, $d$ intersects $d^*$ but is not stabbed by $c^*$, and $d$ intersects the line $l_1$, so $y(c) \le 1$ (since $r(d) \le 2$).
Let $t_{q^+}$ be the line with positive slope tangent to $d^*$ and passing through $q^+$; see Figure~\ref{fig:q_plus_above2}. If $d$ lies below $t_{q^+}$, then $q^+$ is clearly above $t$ and we are done. So assume that $d$ intersects $t_{q^+}$, i.e., the distance between $c$ and $t_{q^+}$ is at most $r(d)$. Since $c$ is below $t_{q^+}$ and $r(d) \le 2$, we conclude that $c$ lies in the strip $U$ defined by $t_{q^+}$ and the line parallel to $t_{q^+}$ at distance 2 from $t_{q^+}$ and below $t_{q^+}$. Moreover, $x(c) \le x(a)$, where $a$ is the intersection point between the lower boundary of $U$ and the line $y=1$ (since $y(c) \le 1$). On the other hand, $c$ is closer to $t_{q^+}$ than to $c^*$, so $c$ must lie on the `right' side of the bisector $P_{t_{q^+}}$ between $c^*$ and $t_{q^+}$, and $c$ is closer to $l_1$ than to $c^*$, so $c$ must also lie on the `right' side of the bisector $P_{l_1}$ between $c^*$ and $l_1$. But, any point in the first quadrant that is on the `right' side of both these bisectors has $x$-coordinate greater than that of $b$, the right intersection point between $P_{t_{q^+}}$ and $P_{l_1}$. Our calculations below show that $x(a) < x(b)$, so $d$ cannot intersect $t_{q^+}$.

Indeed, $t_{q^+}$'s equation is $y=(\frac{61}{50}+\frac{\sqrt{\frac{8663}{3}}}{25})x+\frac{61}{25}+\frac{\sqrt{\frac{8663}{3}}}{50}$, and therefore the equation of the lower boundary of the strip $U$ is $y=(\frac{61}{50}+\frac{\sqrt{\frac{8663}{3}}}{25})x-(\frac{61}{25}+\frac{\sqrt{\frac{8663}{3}}}{50})$ and $x(a) =  \frac{\sqrt{25989}}{83} - \frac{50}{83} < 1.34$.
The bisector between $c^*$ and $t_{q^+}$ is the parabola $P_{t_{q^+}}$:

 \[ x^2 + y^2 = \frac{((\frac{61}{50} + \frac{\sqrt{\frac{8663}{3}}}{25})x - y + (\frac{61}{50} + \frac{\sqrt{\frac{8663}{3}}}{50}))^2}{(\frac{61}{50} + \frac{\sqrt{\frac{8663}{3}}}{25})^2 + 1} \ , \]
and the bisector between $c^*$ and $l_1$ is the parabola $P_{l_1}$: $x^2 - 2y - 1 = 0$.
Finally,

$ x(b) = \frac{1}{283}\left(50 + \sqrt{25989} + \sqrt{2(54289 + 50\sqrt{25989})}\right) > 1.99 $,
so clearly $x(a) < x(b)$.
\end{proof}

First we prove that $d^{l+}$ does not intersect $d$. Let $(a,b)$ denote the center of $d^{l+}$; notice that $a\le 0$ and $b\ge 0$. Since $d^{l+}$ is tangent to $\ell_1$ and has $p^l=(-0.5,0)$ and $q^+=(-0.5,1.83)$ on its boundary, the following equations hold:
\begin{enumerate}[leftmargin=1.5\parindent]
	\item $r(d^{l+})=b+1$,
	\item $(a + 0.5 )^2 + b^2=  r(d^{l+})^2$,
	\item $(a+0.5)^2 + (b - 1.83)^2 =  r(d^{l+})^2$.
\end{enumerate}

These equations solve to: $ a = - \frac{1}{2} - \frac{ \sqrt{283 } }{10}, b= \frac{183}{200}.$

	\begin{figure}[hbt]
		\centering
		\includegraphics[scale=0.55]{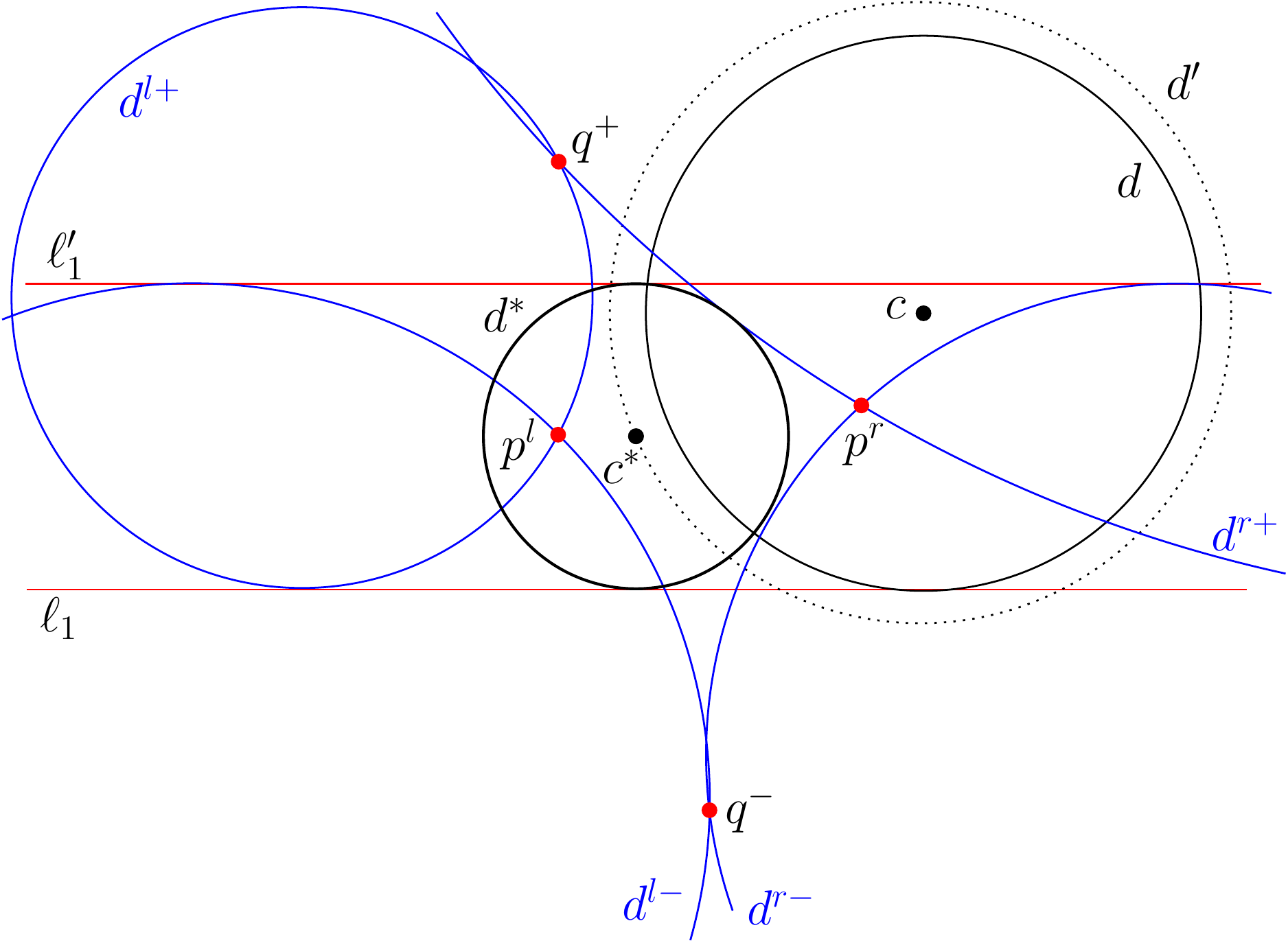}
		\caption{Illustration of $d$, $d^{r+}$, $d^{l+}$, $d^{l-}$, and $d^{r-}$ when $\alpha> 17^{\circ}$ and $y(c) \ge 0$.}
		\label{fig:SmallerThan2Case2}
	\end{figure}

Let $d'$ be the disk of radius 2 that has $c^*$ on its boundary and has its center $c'$ on the ray from $c^*$ to $c$,
i.e., its center is on a line through the origin of angle $\alpha $ with the $x$-axis.
We claim that if $d^{l+}$ does not intersect $d'$, then it also does not intersect $d$.
To verify this claim we consider two cases:
(i) $c$ is on the line segment $c^*c'$
(ii) $c'$ is on the line segment $c^*c$.
In case (i) our claim holds because $d'$ contains $d$.
In case (ii) observe that $c$ is located below the parabola $x^2 - 2 y =1$ (because otherwise $c$ is closer to $c^*$ than to $\ell_1$).
This also implies that the angle $\alpha$ is less than $30^{\circ}$ (because $y(c) \leq 1$).
Therefore, the center of $d^{l+}$ is closer to $c'$ than to $c$, and since the $r(d') \geq r(d)$ our claim holds.

Let $d''$ be the disk of radius 2 that has $c^*$ on its boundary and has its center $c''=(\sqrt{3},1)$,
i.e., its center is on a line through origin that makes angle $30^{\circ}$ with the $x$-axis.
We claim that if $d^{l+}$ does not intersect $d''$, then it also does not intersect $d'$. See Figure~\ref{fig:SmallerThan2Case22} for illustration.
To verify this claim we notice that the center of $d^{l+}$ is closer to $c''$ than to $c' $
(the bisector between $c''$ and $c'$  goes through $c^*$ and thus the center of $d^{l+}$ is closer to $c''$),
and thus our claim holds.

\begin{figure}[hbt]
	\centering
	\includegraphics[scale=0.5]{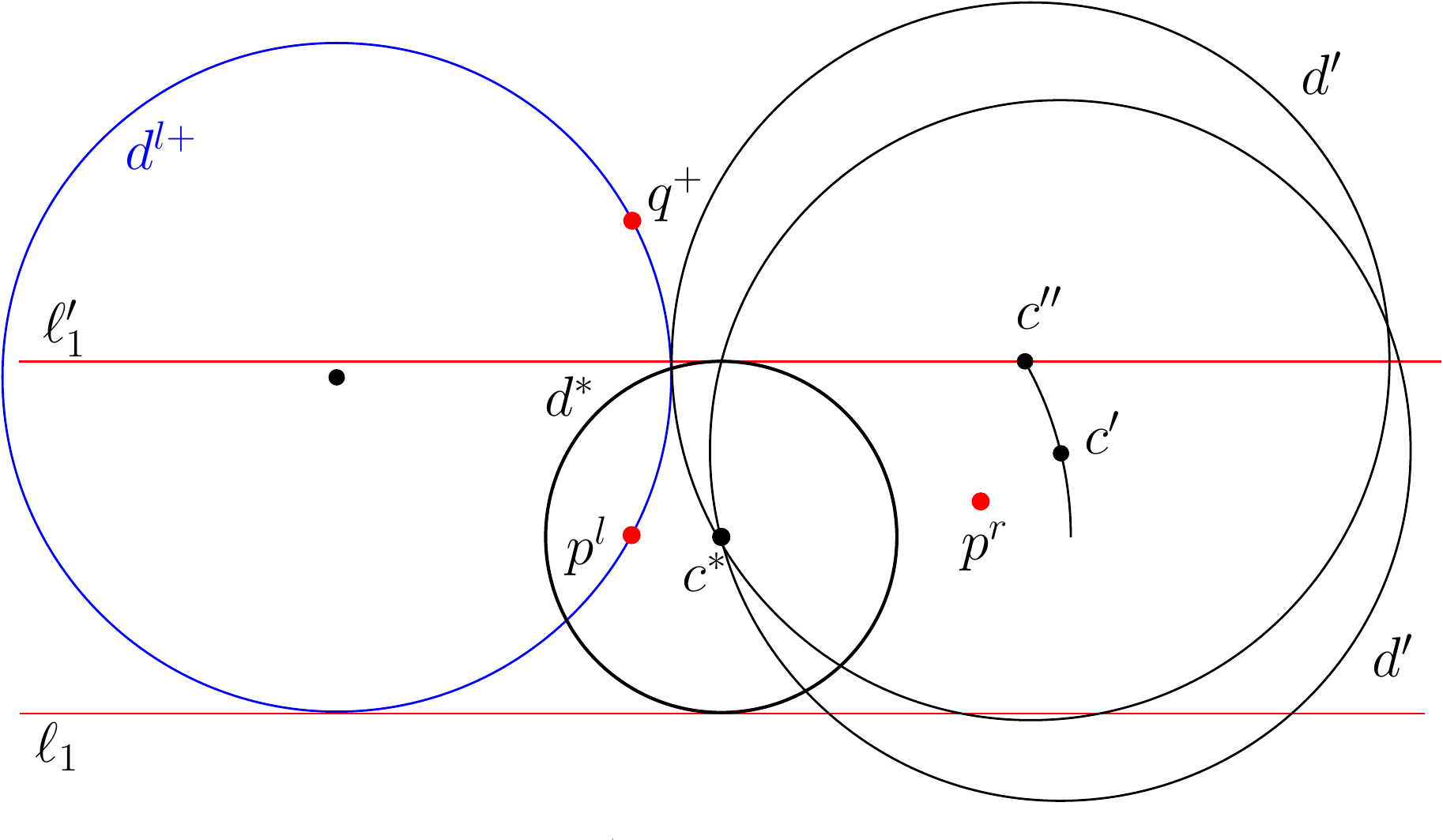}
	\caption{Illustration of $d'$, $d''$ and  $d^{l+}$, when $\alpha\geq 17^{\circ}$.}
	\label{fig:SmallerThan2Case22}
\end{figure}

By the above claims, it suffices to show that $d^{l+}$ does not intersect $d''$,
that is, the distance between their centers more than the sum of their radii:
$ (a - \sqrt{3})^2 + (b -1)^2 > (2 + (b+1) )^2$. By plugging $a$ and $b$ we get:
	\[  \left(- \frac{1}{2} - \frac{ \sqrt{283 } }{10} - \sqrt{3}   \right)^2 + \left( \frac{183}{200} - 1  \right)^2  > \left( 2 + \frac{183}{200} + 1 \right)^2.\]
Therefore, $d^{l+}$ does not intersect $d''$.

Now we prove that $d^{l-}$ does not intersect $d$.
Let $(a,b)$ denote the center of $d^{l-}$; notice that $a\le 0$ and $b\le 0$.
Since $d^{l-}$ is tangent to $\ell'_1$ and has $p^l=(-0.5,0)$ and $q^-=(0.5, -2.5)$ on its boundary, the following equations hold:
\begin{enumerate}[leftmargin=1.5\parindent]
	\item $r(d^{l-})=|b-1|$,
	\item $(a + 0.5 )^2 + b^2=  r(d^{l-})^2$,
	\item $(a-0.5)^2 + (b +2.5)^2 =  r(d^{l-})^2$.
\end{enumerate}

These equations solve to: $ a = - \frac{9}{10} - \frac{ \sqrt{ {203}/{2} } }{5}, b= - \frac{161}{100} - \frac{ \sqrt{ 406 } }{25}.$

	Let $d''$ be a disk of radius 2 that has $c^*$ on its boundary and its center $c''=(2,0)$.
	First we show that if $d^{l-}$ does not intersect $d''$ it does not intersect $d$.
	let $d'$ be a disk of radius 2 tangent to $c^*$ whose center is on the ray from $c^*$ through $c$.
	Notice that if $d^{l-}$ intersects $d$ it must intersect $d'$. To see this consider the two cases with respect to the location of $c$.
	If $c$ is between $c^*$ and $c'$ (i.e., $c \in \overline{c^* c}$), then $d \subseteq d'$.
		Else ($c \notin \overline{c^* c}$) since $r(c) \leq r(c')=2$, $y(c) \leq 1$ and $(a,b)$ (the center of $d^{l-}$) is closer to $c'$ than to $c$, we have that
			if $d^{l-}$ intersects $d$ it must intersect $d'$.
	Now consider the bisector $bis(c',c'')$ between $c'$ and $c''$ and observe that $bis(c',c'')$ goes through $c^*$ and above $(a,b)$.
	The latter follows from the fact the angle between $bis(c',c'')$ and the $x$-axis is at most $15^{\circ}$ (since the ray from $c^*$ to $c$ is of angle at most $30^{\circ}$).

	Now we show that $d^{l-}$ does not intersect $d''$, that is the distance between their centers is more than the sum of their radii:
	$ (a - 2)^2 + (b - 0)^2 > (2 + (|b-1|) )^2$.
	%

	\[  \left(- \frac{9}{10} - \frac{ \sqrt{{203}/{2} } }{5} -2 \right)^2 + \left(- \frac{161}{100} - \frac{ \sqrt{406} }{25}\right)^2 >
	\left(2 +\frac{161}{100} + \frac{ \sqrt{ 406 } }{25} + 1 \right)^2.\]

	Therefore, $d^{l-}$ does not intersect $d''$.

	Now we prove that $d^{r+}$ does not intersect $d^*$. Let $(a,b)$ denote the center of $d^{r+}$; notice that $a > 1$ and $b > 1$.
	Since $d^{r+}$ is tangent to $\ell_1$ and has $p^r=(1/2 + 2\sqrt{6}/5,1/5)$ and $q^+=(-0.5, 1.83)$ on its boundary, the following equations hold:
	\begin{enumerate}[leftmargin=1.5\parindent]
		\item $r(d^{r+})=b+1$,
		\item $(a - (1/2+2\sqrt{6}/5))^2 + (b-1/5)^2=  r(d^{r+})^2$,
		\item $(a+0.5)^2 + (b -1.83)^2 =  r(d^{r+})^2$.
	\end{enumerate}

	These equations solve to:

	\begin{align}
	\noindent \notag a&= \frac{10075 + 5660 \sqrt{6} + \sqrt{8490 \left(46169 + 8000 \sqrt{6} \right)} }{8150}  \  > 5.836    \qquad    (a \approx 5.83662)\\
	\notag b&= \frac{13292307 +3224000 \sqrt{6} +960 \sqrt{ 1415 \left(46169 + 8000 \sqrt{6} \right) }+ 400 \sqrt{8490 \left(46169 + 8000 \sqrt{6}\right)  } }{5313800}\\
	\notag &   \  <  7.51   \qquad  (b \approx 7.50912)
	\end{align}

	Now we show that $d^{r+}$ does not intersect $d^*$, that is the distance between their centers is more than the sum of their radii.
	\[ (a - 0)^2 + (b - 0)^2 > (1 + (b+1) )^2 \]
		\[ a^2 - 4 - 4b > 0  \]
		\[ 5.836^2 -4 - 4* 7.51  > 0 \]




	Consider disk $d^{r-}$, and let $(a,b)$ be its center  (recall that $a > 1 $ and $b < -1$).
	Since $d^{r-}$ is tangent to $\ell'_1$ and has $p^r=(1/2 + 2\sqrt{6}/5,1/5)$ and $q^-=(0.5, -2.5)$ on its boundary, the following equations hold:
	\begin{enumerate}[leftmargin=1.5\parindent]
		\item $r(d^{r-})= |b-1|$,
		\item $(a - \frac{1}{2} - \frac{2 \sqrt{6}}{5} )^2 + (b - 1/5)^2=  r(d^{r-})^2$,
		\item $(a - 0.5)^2 + (b + 2.5)^2 =  r(d^{r-})^2$.
	\end{enumerate}
	Solving these equations we get that
	\[ a =  \frac{1}{54} (27 + 28 \sqrt{6} + 2 \sqrt{2310}) > 3.52 , \qquad  b=  \frac{-1393}{972}  - \frac{8 \sqrt{385}}{243} > -2.08  . \]

	Now we show that $d^{r-}$ does not intersect $d^*$, that is the distance between their centers is more than the sum of their radii.

	\[  ( a - 0 )^2 + (b - 0)^2  > ( 1  - (b  - 1) )^2 \]
	\[  a^2   - 4 + 4b > 0 \]
  \[    3.52^2 - 4  + 4 *(-2.08) >0 \]





\subsubsection{$ \alpha > 17^{\circ}$ and $y(c)< 0$}
\label{alpha-big-y-negative}



	This case is identical under reflection with respect to the $x$-axis to Section~\ref{alpha-big-y-positive}.


%
%
\subsection{Proof of Case $2<\rmin<4$}
\label{secDminTwoToFour}

Recall the disk $\dmin\in D^-$ of radius $\rmin$. Since $\rmin>2$, any disk in $D^-$ has radius larger than $2$, and since $\rmin<4$, $\dmin \in D^-_{\le k}$, for any $k\ge 4$.
In each of the following three claims, we pick a disk $d'$ from a set $D^-_{\le k}$, for some $k \ge 4$. In these claims we \textbf{do not} assume the {\color{mycolor}base setting}; rather we assume, after a suitable rotation, that the center of $d'$ lies on the positive $x$-axis.
In this section we slightly abuse the notation and for a disk $e$ we denote by $x(e)$ and $y(e)$ the $x$-coordiante and $y$-coordinate of the center of $e$, respectively.

\begin{claim}
	\label{cl:R5moreThanHalf}
	Let $d'$ be the disk in $D^{-}_{ \leq 5}$ with maximum $\delta(d')$, and assume that its center lies on the positive $x$-axis.
	If $\delta(d') \geq 0.5$, then $\{(0,0), (2,0), (0.4, 2), (0.4, -2)\}$ stabs $D$.
\end{claim}
\begin{proof} Any disk in $D\setminus D^-$ contains $(0,0)$. Thus, we prove the claim for $D^-$. We claim for any disk $e \in D^{-}$ that if $y(e)\ge 0$ then $e$ contains
	$(2,0)$ or $(0.4, 2)$, otherwise $e$ contains
	$(2,0)$ or $(0.4, -2)$.
	We prove our claim for $y(e)\ge 0$, since the proof for $y(e) < 0$ is symmetric. Notice that the center of $d'$ is $(r(d') + \delta(d'),0)$. Depending on the sign of $x(e)$ we consider two cases.

	\paragraph*{$x(e) \le 0$.}
		Let $d_{2}$ be the disk of radius 2, with negative $x(d_2)$, that has $(0, 0)$ and $(0.4, 2)$ on its boundary.	If $(x_2,y_2)$ denotes the center of $d_{2}$, then we have the following equations
		\[ \qquad  x_{2}^2 + y_{2}^2 = 4, \qquad  (x_2 -0.4)^2 +(y_2 - 2)^2 = 4,\]
		which solve to  $x_2 = \frac{1}{5} - \sqrt{ {37}/{13}}$ and $y_2 = 1 + \frac{ \sqrt{ {37}/{13}}}{5}$.

		Since $r(d') \le 5$ and $0.5 \le \delta(d') \le 1$, it is easy to verify that the  distance between the centers of $d_2$ and $d'$ is greater than the sum of their radii, that is
		\[\left(x_2 - (r(d') + \delta(d')) \right)^2 + (y_2 - 0)^2 > (2+ r(d'))^2.\]
		Therefore, $d_2$ does not intersect $d'$.  We use this fact to show below that any disk $e \in D^-$ with $y(e) \ge 0$ and $x(e) \leq 0$ (that
		intersects both $d'$ and $d^*$) must contain the point $q = (0.4,2)$.

\begin{figure}[hbt]
	\centering
	\includegraphics[scale=0.6]{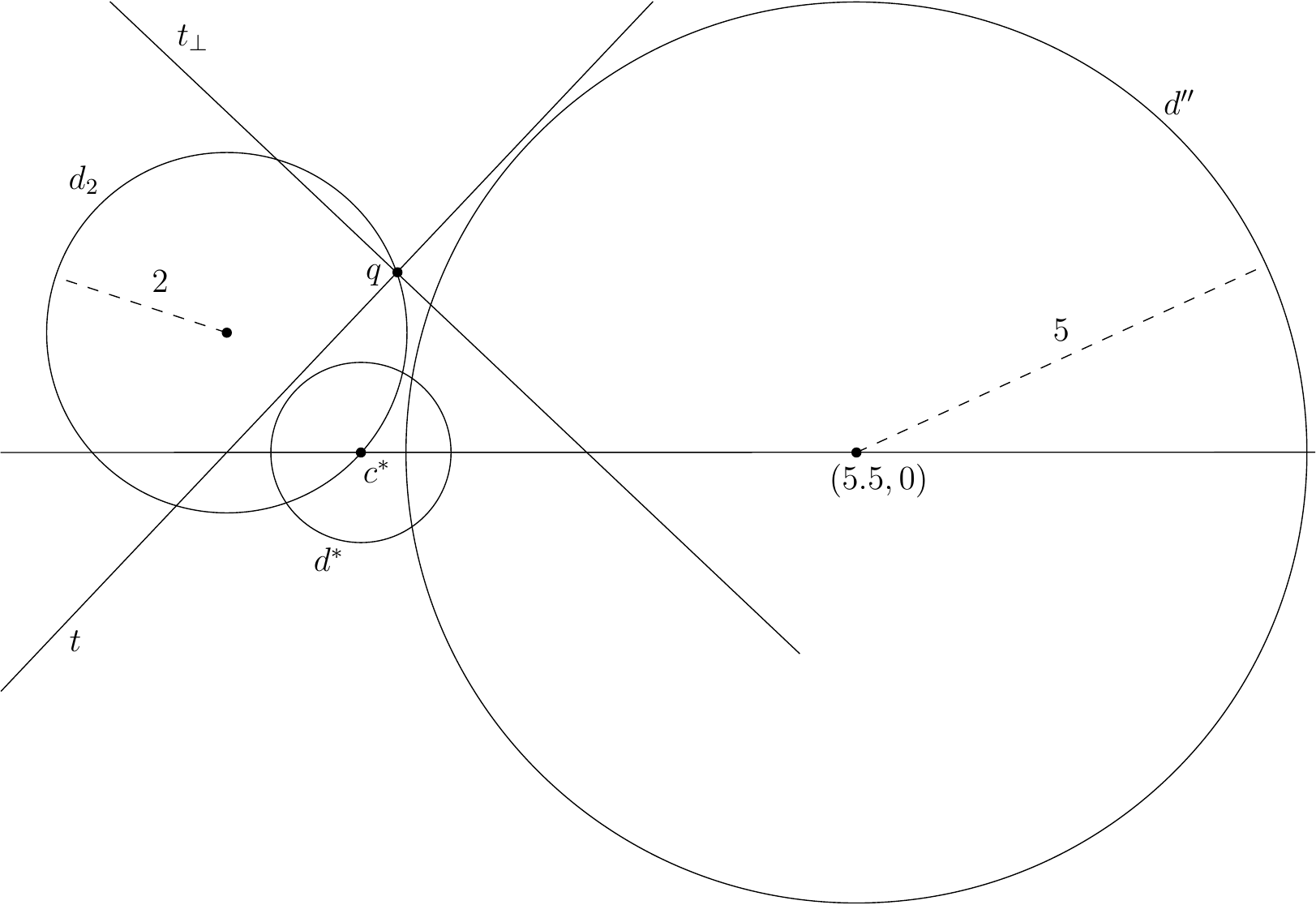}
	\caption{Proof of Claim~\ref{cl:R5moreThanHalf}, $c(e)$ in the second quadrant.}
	\label{fig:r_min2-4_claim1}
\end{figure}

		Let $d''$ be the disk of radius 5 centered at $(5.5,0)$ (see Figure~\ref{fig:r_min2-4_claim1}). Notice that since $e$ intersects $d'$ it must also intersect $d''$.
		Let $t'$ be the (relevant) mutual tangent to $d''$ and $d^*$, i.e., $t': y = \frac{8}{\sqrt{57}}x + \frac{11}{\sqrt{57}}$. It is easy to verify that the point $q$ lies above $t'$. Now, let $t$ be the line parallel to $t'$ that passes through $q$, let $t_{\bot}$ be the line perpendicular to $t$ that passes through $q$, and let $P=P(q,d^*)$ be the bisector between $q$ and $d^*$.

		We first observe that in the halfplane to the left of the $y$-axis, $t_{\bot}$ is above  $P$. This can be verified by showing that for any point $p$ on $t_{\bot}$ with $x(p) \le 0$, we have that the distance between $p$ and $q$ is less than the distance between $p$ and $d^*$ (which is the distance between $p$ and $c^*$ minus 1).

	 If $c(e)$ is above $t_{\bot}$ then clearly $e$ contains $q$ (since $e$ intersects $d^*$), so assume that $c(e)$ is below $t_{\bot}$. Our assumptions on $e$ ($e \in D^-$, $r(e) > 2$, $c(e)$ is in the second quadrant) imply that $c(e)$ is above $t$.  Now, if $e$ intersects $t$ at a point above $q$, then by the triangle inequality $e$ contains $q$. Otherwise, since $e$ intersects $d''$, $e$ must intersect the clockwise arc of $d_2$ from $c^*$ to $q$. But, then $e$ contains $q$ (since $r(e) > r(d_2) =2$ and $c^* \notin e$).

		\paragraph*{$x(e) > 0$.}
		Let $d_{2}$ be the disk of radius 2, with $x(d_2) > 0.4$, that has $(2, 0)$ and $(0.4, 2)$ on its boundary.
		If $(x_2,y_2)$ denotes the center of $d_{2}$, then we have the following equations
		\[ \qquad  (x_{2} - 2)^2 + y_{2}^2 = 4, \qquad  (x_2 -0.4)^2 +(y_2- 2)^2 = 4,\]
		which solve to  $x_2 = \frac{6}{5} +  \sqrt{ {59}/{41}}$ and $y_2 = 1 + \frac{ 4 \sqrt{ {59}/{41}}}{5}$.

		The distance between the centers of $d_2$ and $d^*$ is greater than the sum of their radii, that is $(x_2)^2 + (y_2 )^2 > (2 + 1)^2$,
		which implies that $d_2$ does not intersect $d^*$. This, together with some of our geometric observations, implies that $e$ contains point $(2,0)$ or point $(0.4, 2)$.
		To see this, let $l$ be the line through the points $(2,0)$ and $(0.4, 2)$.
		Let $V$ be the strip that is defined by the two lines perpendicular to $l$ through $(2,0)$ and $(0.4,2)$, respectively.
		Notice that $d^*$ is contained in $V$ below $l$.
    If $c(e)$ is outside $V$ and above $l$, then by Observation~\ref{outside-strip} we have that $e$ contains $(2,0)$ or $(0.4,2)$.
    If $c(e)$ is in $V$ and above $l$, then by Observation~\ref{radius} we have that $e$ contains $(2,0)$ or $(0.4,2)$.
		Finally, if  $c(e)$ is below $l$ (recall that $c(e)$ is in the first quadrant), then, since $r(e) > 2$ and $c^* \notin e$, $e$ contains $(0.4,2)$.
		%
\end{proof}

\begin{claim}
	\label{cl:moreThanHalf}
	Assume that $D^{-}_{ \leq 5} $ does not contain any disk with $\delta(\cdot)$ greater or equal to $0.5$. Let $d'$ be the disk in $D^{-}_{ \leq 20}$ with maximum $\delta(d')$,
	and assume that its center lies on the positive $x$-axis.
	If $\delta(d') \geq 0.5$, then $\{(0,0), (2,0), (-0.15, 2.7), (-0.15, -2.7)\}$ stabs $D$.
\end{claim}

\begin{proof}

Any disk in $D\setminus D^-$ contains $(0,0)$. Thus, we prove the claim for $D^-$. We claim for any disk $e \in D^{-}$ that if $y(e)\ge 0$ then $e$ contains
  $(2,0)$ or $(-0.15, 2.7)$, otherwise $e$ contains
	$(2,0)$ or $(-0.15, -2.7)$.
	We prove our claim for $y(e)\ge 0$, since the proof for $y(e) < 0$ is symmetric. Notice that the center of $d'$ is $(r(d') + \delta(d'),0)$. We consider the following three cases.

\paragraph*{\item $x(e) \le 0$.}
    Let $d_{2}$ be the disk of radius 2, with negative $x(d_2)$, that has $(0, 0)$ and $(-0.15, 2.7)$ on its boundary.
		If $(x_2,y_2)$ denotes the center of $d_{2}$, then we have the following equations
		\[ \qquad  x_{2}^2 + y_{2}^2 = 4, \qquad  (x_2 + 0.15)^2 +(y_2 - 2.7)^2 = 4,\]
		which solve to  $x_2 = - \frac{3}{40} - \frac{9 \sqrt{ {139}/{13}}}{20}$ and $y_2 = \frac{27}{20} + \frac{ \sqrt{ {139}/{13}}}{40}$.

		Since $r(d') \le 20$ and $0.5 \le \delta(d') \le 0.5$,
		it is easy to verify that the  distance between the centers of $d_2$ and $d'$ is greater than the sum of their radii,
		that is  $(x_2 - (r(d') + \delta(d')) )^2 + y_2^2 > (2 + r(d'))^2$.

		Therefore, $d_2$ does not intersect $d'$.  As above, we use this fact to show that any disk $e \in D^-$ with $y(e) \ge 0$ and $x(e) \leq 0$ (that
		intersects both $d'$ and $d^*$) must contain the point $q= (-0.15,2.7)$.

			Let $d''$ be the disk of radius 20 centered at $(20.5,0)$. Notice that since $e$ intersects $d'$ it must also intersect $d''$.
		Let $t'$ be the (relevant) mutual tangent to $d''$ and $d^*$, i.e., $t': y = \frac{38}{\sqrt{237}}x + \frac{41}{\sqrt{237}}$. It is easy to verify that the point $q$ lies above $t'$. Now, let $t$ be the line parallel to $t'$ that passes through $q$, let $t_{\bot}$ be the line perpendicular to $t$ that passes through $q$, and let $P=P(q,d^*)$ be the bisector between $q$ and $d^*$.

		We first observe that in the halfplane to the left of the $y$-axis, $t_{\bot}$ is above  $P$. This can be verified by showing that for any point $p$ on $t_{\bot}$ with $x(p) \le 0$, we have that the distance between $p$ and $q$ is less than the distance between $p$ and $d^*$ (which is the distance between $p$ and $c^*$ minus 1).

	 If $c(e)$ is above $t_{\bot}$ then clearly $e$ contains $q$ (since $e$ intersects $d^*$), so assume that $c(e)$ is below $t_{\bot}$. Our assumptions on $e$ ($e \in D^-$, $r(e) > 2$, $c(e)$ is in the second quadrant) imply that $c(e)$ is above $t$.  Now, if $e$ intersects $t$ at a point above $q$, then by the triangle inequality $e$ contains $q$. Otherwise, since $e$ intersects $d''$, $e$ must intersect the clockwise arc of $d_2$ from $c^*$ to $q$. But, then $e$ contains $q$ (since $r(e) > r(d_2) =2$ and $c^* \notin e$).

		\paragraph*{$x(e) > 0$ and $r(e) \ge 5$.}
		Let $d_{5}$ be the disk of radius 5, with positive $x(d_5)$, that has $(2, 0)$ and $(-0.15, 2.7)$ on its boundary.
		If $(x_5,y_5)$ denotes the center of $d_{5}$, then we have
		\[ \qquad  (x_{5} - 2)^2 + y_{5}^2 = 25, \qquad  (x_{5} + 0.15)^2 +(y_{5} - 2.7)^2 = 25,\]
		which solve to  $x_5 = \frac{37}{40} + \frac{ 243 \sqrt{ {87}/{953}}}{20}$ and $y_5 = \frac{27}{20} + \frac{ 387 \sqrt{ {87}/{953}}}{40}$.

		Since the distance between the centers of $d_5$ and $d^*$ is greater than the sum of their radii, i.e., $(x_5)^2 + (y_5 )^2 > (1 + 5)^2$, these disks do not intersect.
		Thus, using observations~\ref{outside-strip} and~\ref{radius} as above, we have that $e$ contains point $(2,0)$ or point $(-0.15, 2.7)$.

		\paragraph*{$x(e) > 0$ and $r(e) < 5$.}
		Recall that $r(e) > 2$. Since $e\in D^-_{\le 5}$, by the claim's assumption we have that $\delta(e) < 0.5$.
		Using an idea similar to the one used in the previous case, we show that $e$ contains $(2,0)$ or $(-0.15, 2.7)$.
		Let $d_{2}$ be the disk of radius 2, with positive $x(d_2)$, that has $(2, 0)$ and $(-0.15, 2.7)$ on its boundary.
		If $(x_2,y_2)$ denotes the center of $d_{2}$, then
		\[ \qquad  (x_{2} - 2)^2 + y_{2}^2 = 4, \qquad  (x_{2} + 0.15)^2 +(y_{2} - 2.7)^2 = 4,\]
		which solve to  $x_2 = \frac{37}{40} + \frac{ 27 \sqrt{ {327}/{953}}}{20}$ and $y_2 = \frac{27}{20} + \frac{ 43 \sqrt{ {327}/{953}}}{40}$.

		We claim that $d_{2}$ does not intersect the disk of radius $0.5$ and center $c^*$, which we denote by $d_{0.5}$.
		This claim can be verified by showing that the distance between the centers of these disks is greater than the sum of their radii,
		that is, $x_{2}^2 + y_{2}^2 > (0.5 + 2)^2$.
		Thus, again, using observations~\ref{outside-strip} and~\ref{radius} we have that $e$ contains point $(2,0)$ or point $(-0.15, 2.7)$.	
     %
\end{proof}


The proofs of the following two claims use similar arguments and can be found in Section~\ref{sec:app_B}.

\begin{claim}
	\label{cl:R5moreThan02}
	Assume that $D^{-}_{ \leq 20} $ does not contain any disk with $\delta(\cdot)$ at least  $0.5$.
	Let $d'$ be the disk in $D^{-}_{ \leq 5}$ with maximum $\delta(d')$, and assume that its center lies on the positive $x$-axis.
	If $0.11  \le \delta(d') < 0.5$, then $\{(0,0), (2,0), (-0.15, 1.75), (-0.15, -1.75)\}$ stabs $D$.
\end{claim}

Claims \ref{cl:R5moreThanHalf}, \ref{cl:moreThanHalf}, and \ref{cl:R5moreThan02} imply that if $D^-_{\le 5}$ contains a disk with $\delta(\cdot) \ge 0.11$, or if $D^-_{\le 20}$ contains a disk with $\delta(\cdot) \ge 0.5$, then there exists a set of four points that stabs $D$. It remains to prove our claim for the case where every disk in $D^-_{\le 5}$ has $\delta(\cdot) < 0.11$ and every disk in $D^-_{\le 20}$ has $\delta(\cdot) < 0.5$. To that end, we assume the {\color{mycolor} base setting}.


In what follows we show that in this case $D$ is stabbed by $(0,0),(2.5,1),(-2.5,1),$ and $(0,-1.52)$.
Here we assume w.l.o.g. that  one of the tangents ($\ell_1$) is as described in Section~\ref{secDminBig}. See Figure~\ref{fig:fin}.

\begin{figure}[hbt]
	\centering
	\includegraphics[scale=0.4]{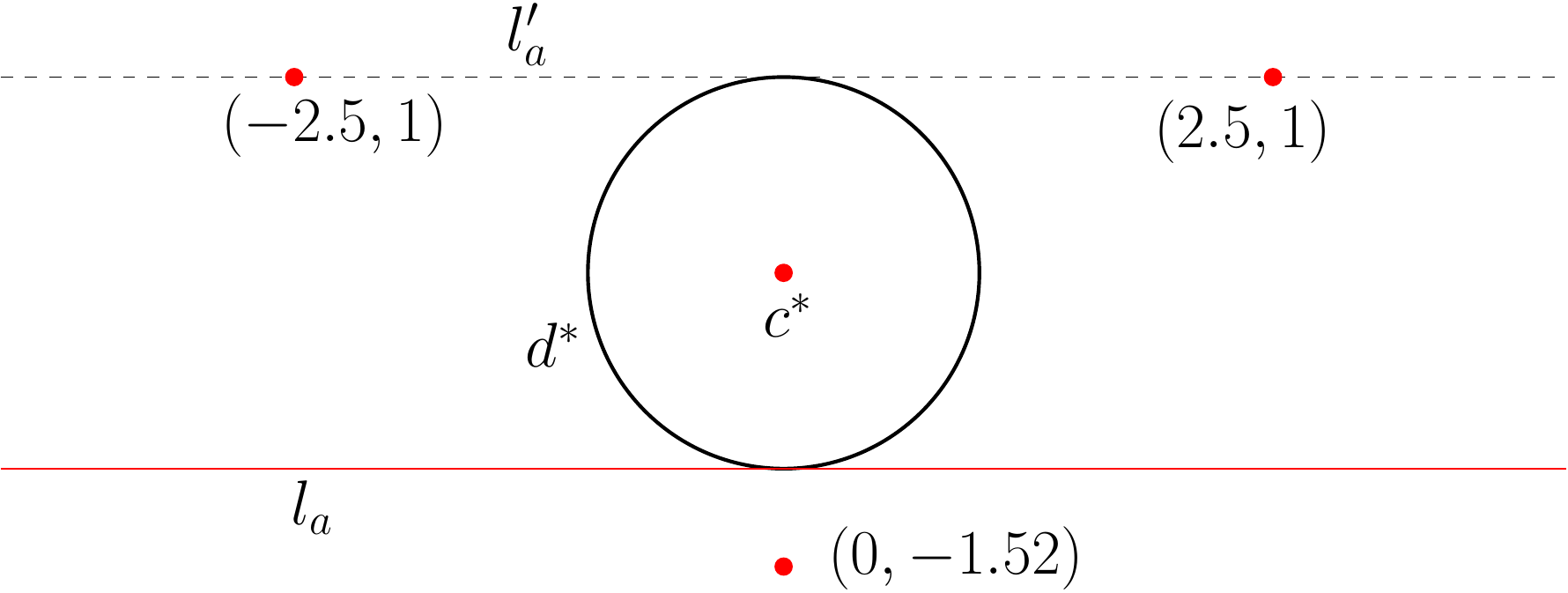}
	\caption{The 4 piercing points in this case are depicted in red.}
	\label{fig:fin}
\end{figure}

\begin{claim}
\label{cl:R20lessThan05}
	If for every disk  $e_1 \in D^{-}_{ \leq 5}$ we have $\delta(e_1) < 0.11$ and for every disk $e_2 \in D^{-}_{ \leq 20}$ we have $\delta(e_2) < 0.5$, then
	the set $\{(0,0),(2.5,1),(-2.5,1),(0,-1.52)\}$ stabs $D$.
\end{claim}



\section{Geometric Observations}
\label{sec:geom_observations}


\begin{figure}
	\begin{center}
		\includegraphics[width=.5\textwidth]{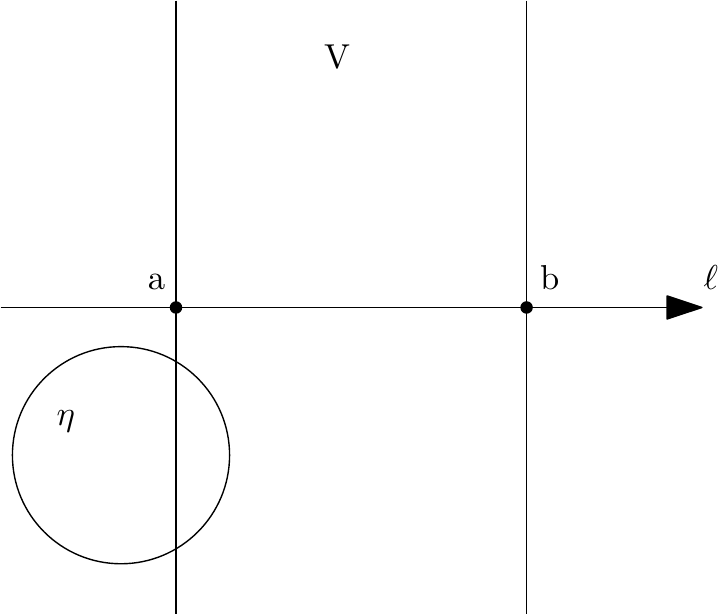}
	\end{center}
	\caption{The setting for Observations 1-3.}
	\label{fig:V}
\end{figure}

In the following three observations, $\ell$ is a line and $a$ and $b$ are two points on $\ell$.
We think of $\ell$ as a directed line, such that $a$ precedes $b$ along $\ell$. If an object lies to the left (right) of $\ell$ (when traversing $\ell$),
then we say that it lies ``above'' (``below'') $\ell$.
Let $V$ be the strip that is defined by the two lines perpendicular to $\ell$ through $a$ and $b$, respectively.
Finally, let $\eta$ be a disk with center below $\ell$, such that (i) $\eta \cap \ell = \emptyset$ or $\eta$ is tangent to the line segment $ab$,
and (ii) $\eta \cap V \neq \emptyset$; see Figure~\ref{fig:V}.

\begin{Observation}\label{obs:obs1}
	Let $\epsilon$ be a disk with center above $\ell$.
	If $\epsilon \cap (\eta \cap V) \neq \emptyset$, then $\epsilon$ intersects the segment $ab$.
\end{Observation}

\begin{Observation}\label{outside-strip}
	Let $\epsilon$ be a disk with center above $\ell$ and not in $V$.
	If $\epsilon \cap (\eta \cap V) \neq \emptyset$, then $\epsilon$ contains $a$ or $b$.
\end{Observation}

\begin{Observation}\label{radius}
	Let $\delta$ be a disk tangent to $a$, $b$, and $\eta$,
	and let $\epsilon$ be a disk with center above $\ell$ and in $V$ and with radius at least the radius of $\delta$.
	If $\epsilon \cap \eta \neq \emptyset$, then $\epsilon$ contains $a$ or $b$.
\end{Observation}

\old{
	\begin{lemma}\label{paf}
		Let
		\begin{itemize}
			\item $\ell$ be a horizontal line;
			\item $a$ and $b$ be distinct points on $\ell$ with $a$ to the left of $b$;
			\item $V$ be the vertical strip whose left side contains $a$ and whose right side contains $b$;

			\item $\eta$ be a disk with center below $\ell$, such that  $c \cap \ell = \emptyset$ or $\eta$ tangent to the line segment $ab$  and
			$\eta \cap V \neq \emptyset$;
			\item $e $ be a disk with center above $\ell$ that intersects $V\cap \eta$.
		\end{itemize}
		Then $e$ intersects the segment $ab$.
	\end{lemma}

	\begin{corollary}\label{outside-strip}
		Let
		\begin{itemize}
			\item $\ell$, $a$, $b$, $\eta$, and $V$ be defined as in Lemma~\ref{paf};
			\item $e$ be a disk with center above $\ell$, not in $V$, and that intersects $V\cap c$.
		\end{itemize}
		Then $e$ contains $a$ or $b$.
	\end{corollary}

	\begin{corollary}\label{radius}
		Let
		\begin{itemize}
			\item $\ell$, $a$, $b$, $\eta$, and $V$ be defined as in Lemma~\ref{paf};
			\item $d$ be disk tangent to $a$, $b$, and $\eta$;
			\item $e$ be a disk with center in $V$  above $\ell$ whose radius is at least the radius of $d$;
			\item $e \cap \eta \neq \emptyset$.
		\end{itemize}
		Then $e$ contains $a$ or $b$.
	\end{corollary}
}  

\begin{figure}
	\begin{center}
		\includegraphics[width=.6\textwidth]{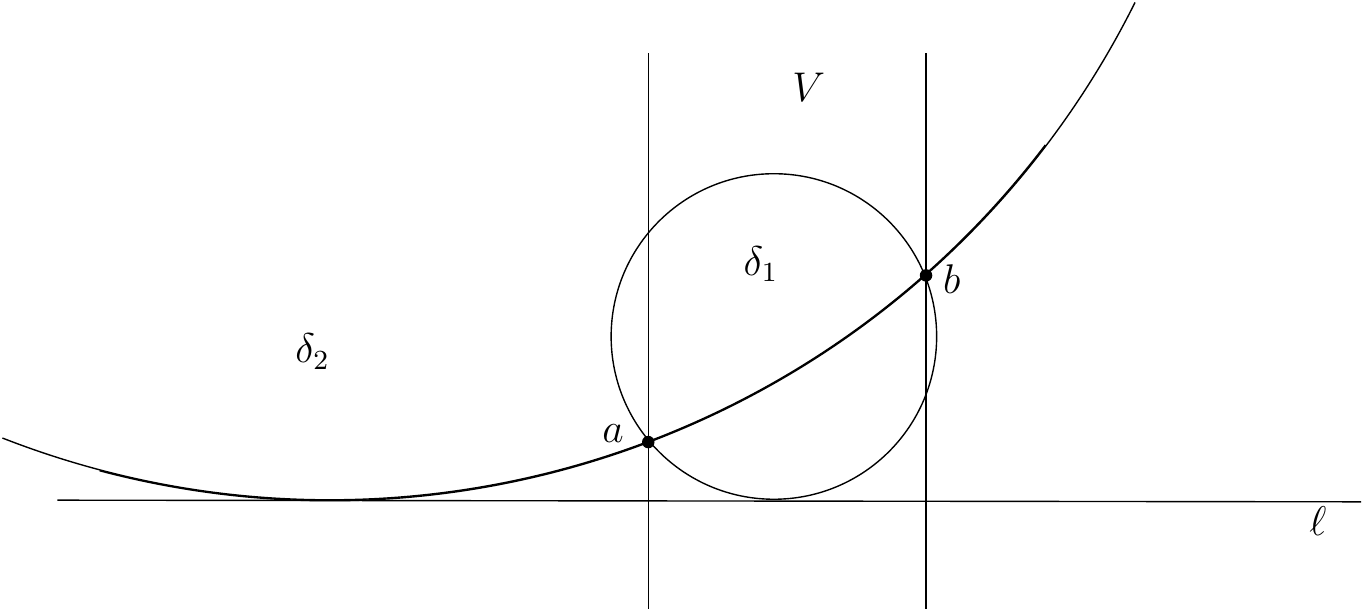}
	\end{center}
	\caption{The setting for Observation~\ref{b}.}
	\label{fig:b}
\end{figure}

\begin{Observation}
	\label{b}
	Consider the setting depicted in Figure~\ref{fig:b}. Let
	\begin{itemize}
		\item $\ell$ be a horizontal line;
		\item $a$ and $b$ be two points above $\ell$ with $a$ to the left of $b$;
		\item $V$ be the vertical strip whose left side contains $a$ and whose right side contains $b$;
		\item $\delta_1$ and $\delta_2$ be the two disks that have $a$ and $b$ on their boundary and are tangent to $\ell$ from above;
		\item $\delta \in \{\delta_1 , \delta_2 \}$;
	\end{itemize}
	If $\epsilon$ is a disk with center in $V$ and above the center of $\delta$ that intersects $\ell$, then $\epsilon$ contains $a$ or $b$.
\end{Observation}

\begin{figure}
	\begin{center}
		\includegraphics[width=.6\textwidth]{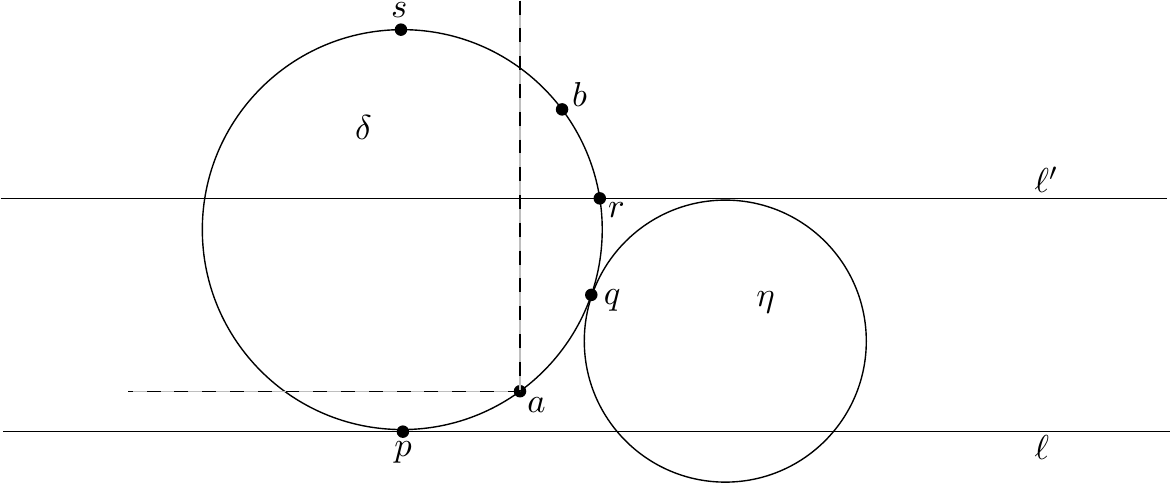}
	\end{center}
	\caption{The setting for Observation~\ref{big}.}
	\label{fig:big}
\end{figure}

The following two observations (i.e., Observations~\ref{big}~and~\ref{Obs:5Plus}) are somewhat more involved, so we present a proof to the former one. The proof to the latter one is similar.

\begin{Observation}
	\label{big}
	Consider the setting depicted in Figure~\ref{fig:big}. Let
	\begin{itemize}
		\item $\ell$ be a horizontal line;
		\item $\eta$ be a disk that intersects $\ell$ and whose center is above $\ell$;
		\item $\ell'$ be the line parallel to $\ell$ and tangent to $\eta$ from above;
		\item $\delta$ be a disk that intersects $\ell'$ and is tangent to $\ell$ from above at some point $p$ and is tangent to $\eta$ from the left at some point $q$;
		\item $r$ be the intersection point of $\ell'$ with the right boundary of $\delta$;
		\item $s$ be the highest point of $\delta$;
		\item $a$ be a point on the boundary of $\delta$ on the counterclockwise arc from $p$ to $q$;
		\item $b$ be a point on the boundary of $\delta$  on the counterclockwise arc from $r$ to $s$.
	\end{itemize}
	If $\epsilon$ is a disk with center in the upper-left quadrant of $a$ that
	intersects both $\ell$ and $\eta$, then $\epsilon$ contains $a$ or $b$.
\end{Observation}
\begin{proof}
	If $\epsilon$ contains $a$ or $b$, then we are done, so assume that both $a$ and $b$ are not in $\epsilon$.
	By Observation~\ref{Obs:Rays}, we have that $\epsilon$ cannot interest  $\ell $ at a point that belongs to the lower-right quadrant of $a$,
	and that $\epsilon$ cannot interest  $\eta$  at a point that belongs to the lower-right quadrant of $a$.
	We distinguish between two cases according to the location of $c(\epsilon)$, the center of $\epsilon$.
	\begin{itemize}
		\item $c(\epsilon)$ is in the upper-right quadrant of $b$ (this is only possible if $b$ is to the left of $a$).
		By Observation~\ref{Obs:Rays}, we have that $\epsilon$ must interest $\ell$ at a point with $x$-coordinate between
		the $x$-coordinate of $a$ and the $x$-coordinate of $b$, but this implies that $\epsilon$ intersects $\delta$ more than
		twice, which is of course impossible.
		\item Otherwise ($c(\epsilon)$ is not in the upper-right quadrant of $b$).
		By Observation~\ref{Obs:Rays}, and since $\epsilon$ interests $\eta$, we have that $\epsilon$ intersects
		$\delta$ at a point on the boundary of $\delta$ on the counterclockwise arc from $a$ to $b$.
		However, since $\epsilon$ needs to intersect $\ell$ to the left of $a$, we have that $\epsilon$ and $\delta$ intersect
		more than twice, which again is impossible.
	\end{itemize}
\end{proof}

\setcounter{theorem}{15}
\begin{corollary}\label{smaller}
	Given the setting for Observation~\ref{big},
	let $\eta'$ be a disk that is contained in $\eta$.
	If $\epsilon$ is a disk with center in the upper-left quadrant of $a$ that
	intersects both $\ell$ and $\eta'$, then $\epsilon$ contains $a$ or $b$.
	\old{
		Let
		\begin{itemize}
			\item $\ell$, $a$, $b$, $\eta$, and $\delta$ be defined as in Observation~\ref{big};
			\item $\eta'$ be a disk that is contained in disk $\eta$;
			\item $\epsilon$ be a disk with center in the upper-left quadrant of $a$ that
			intersects $\ell$ and $\eta'$.
		\end{itemize}
		Then $\epsilon$ contains $a$ or $b$.
	}
\end{corollary}

\begin{figure}
	\begin{center}
		\includegraphics[width=.6\textwidth]{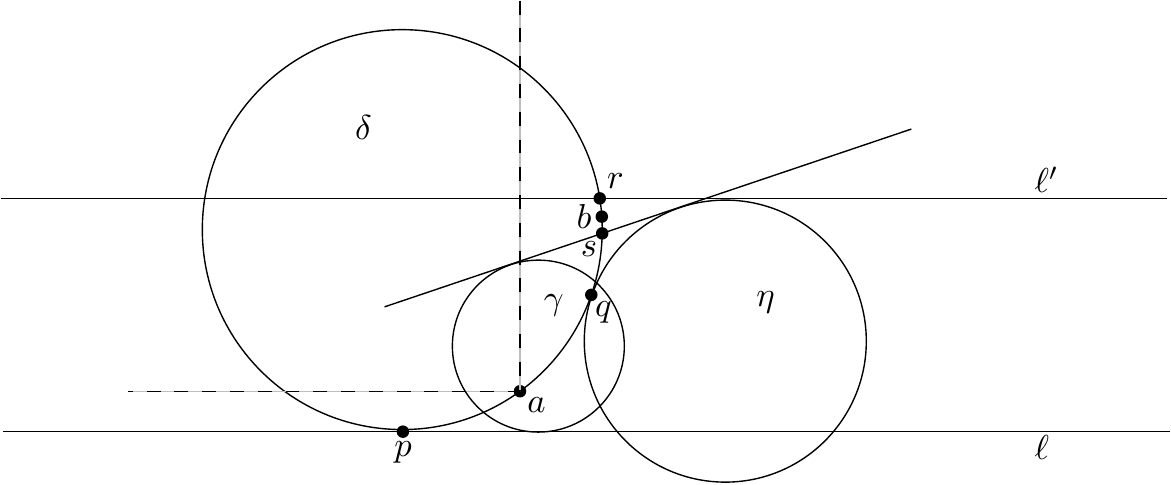}
	\end{center}
	\caption{The setting for Observation~\ref{Obs:5Plus}.}
	\label{fig:5Plus}
\end{figure}

\begin{Observation}\label{Obs:5Plus}
	Consider the setting depicted in Figure~\ref{fig:5Plus}. Let
	\begin{itemize}
		\item $\ell$ be a horizontal line;
		\item $\eta$ be a disk that intersects $\ell$ and whose center is above $\ell$;
		\item $\ell'$ be the line parallel to $\ell$ and tangent to $\eta$ from above;
		\item $\delta$ be a disk that intersects $\ell'$ and is tangent to $\ell$ from above at some point $p$ and is tangent to $\eta$ from the left at some point $q$;
		\item $r$ be the intersection of $\ell'$ with the right boundary of $\delta$;
		\item $\gamma$  be a disk that does not intersect $\ell'$, such that $\gamma$ is tangent to $\ell$ from above and its center is to the left of $\eta$.
		\item $s$ be the right intersection point of $\delta$ and the upper mutual tangent of $\eta$ and $\gamma$.
		\item $a$ be a point on the boundary of $\delta$ on the counterclockwise arc from $p$ to $q$;
		\item $b$ be a point on the boundary of $\delta$ on the counterclockwise arc from $s$ to $r$.

	\end{itemize}
	If $\epsilon$ is a disk with center in the upper-left quadrant of $a$ that
	intersects $\ell$, $\eta$, and $\gamma$, then $\epsilon$ contains $a$ or $b$.
\end{Observation}

\begin{Observation}\label{Obs:Rays}
	Let $p$ be a point in the plane.
	Let $\delta$ be a disk with center $c$ in the upper-left quadrant of $p$.
	If $\delta$ intersects the lower-right quadrant of $p$, then $\delta$ contains $p$.
\end{Observation}
\begin{proof}
	Assume $\delta$ intersects the lower-right quadrant of $p$.
	Then, $\delta$'s boundary intersects the upper boundary or the left boundary of the lower-right quadrant of $p$. Assume w.l.o.g. that $\delta$'s boundary intersects the left boundary at a point $x$.
	We have by the triangle inequality that $|cx| > |cp|$, since $\angle(cpx) > 90^{\circ}$.
	Thus, $\delta$ contains $p$.
\end{proof}




\newpage
\appendix

\section{Proof of Claim~\ref{cl:blabla}}\label{sec:blabla}
In this section, we show that any disk $e \in D$ that is centered at a point of negative $y$-coordinate
and does not contain $p^l$ nor $q^-$ intersects $\ell'_1$.
Assume towards contradiction that there exists such a disk $e$ that does not intersect $\ell'_1$.
We first consider the case where the center of $e$  ($c(e)$) is below  $q^-$ (w.r.t. the $y$-coordinate).
Then, we consider the case where the center of $e$ lies above $q^-$ (and has negative $y$-coordinate).
In both cases we show that $e$ cannot exist.

\begin{claim}
Let $e \in D$ be a disk whose center is below  $q^-$ and does not contain $p^l$ nor $q^-$,
then $e$ intersects $\ell'_1$.
\end{claim}
\begin{proof}

We distinguish between two cases according to the $x$-coordinate of $e$'s center.
If $c(e)$ has negative $x$-coordinate, then let $t$ be the vertical line through the point $(1,0)$.
In this case it is sufficient to show that the disk $e'$ (whose center has negative $x$-coordinate and $y$-coordinate below $q^-$) that is tangent to $t$ and $d^*$ and has the point $q^-$ on its boundary intersects $\ell'_1$, since this implies that $e$ intersects $\ell'_1$.

If $c(e)$ has positive $x$-coordinate, then we consider two sub-cases according to the location of $d$'s center. If $c(d)$ has positive $y$-coordinate, then let $t$ be the vertical line through point $(-1,0)$; otherwise, let $t$ be the line tangent to $d^*$ at a point above the $x$-axis and forming an angle of $30^{\circ}$ with the (positive) $x$-axis. Notice that in the latter sub-case ($c(d)$ has negative $y$-coordinate), the angle between $\ell_2$ and the (positive) $x$-axis is at most $30^{\circ}$, since otherwise $d$ does not intersect $\ell_2$ (as its radius is at most 2). In both sub-cases it is sufficient to show that the disk $e'$ (whose center has positive $x$-coordinate and $y$-coordinate below $q^-$) that is tangent to $t$ and $d^*$ and has the point $q^-$ on its boundary intersects $\ell'_1$, since this implies that $e$ intersects $\ell'_1$.

Let $e'$ be a disk whose center is below $q^-$ that is tangent to $t$ and $d^*$ and has the point $q^-$ on its boundary, and let $(a,b)$ be its center.

First consider the case where $c(e)$ has negative $x$-coordinate.
We have the following equations:
	\begin{enumerate}[leftmargin=1.5\parindent]
		\item $r(e') = | 1 - a | $,
		\item $a^2 + b^2 = (r(e') + 1)^2$, and
		\item $(a - x(q^-))^2 + (b - y(q^-))^2= r(e')^2$.
	\end{enumerate}
	Solving these equations when $q^-=(0,-17/10)$ we get that
    $b > -5.4 $ and $r(e') > 7$.
	Solving these equations when $q^-=(1/2,-5/2)$ we get that
    $b > -4.9 $ and  $r(e') > 5.9$.
  Solving these equations when $q^-=(-1/2,-183/100)$ we get that
    $b > -13 $ and $r(e') > 40$.
	Thus, in all three cases $e'$ intersects $\ell'_1$ ($b + r \geq 1$).

Next, consider the sub-case where $c(e)$ has positive $x$-coordinate and $d$ has positive $y$-coordinate.
We have the following equations:
	\begin{enumerate}[leftmargin=1.5\parindent]
		\item $r(e') = | 1 + a | $,
		\item $a^2 + b^2 = (r(e') + 1)^2$, and
		\item $(a - x(q^-))^2 + (b - y(q^-))^2= r(e')^2$.
	\end{enumerate}
	Solving these equations when $q^-=(0,-17/10)$ we get that
    $b > -5.4 $ and $r(e') > 7$.
	Solving these equations when $q^-=(1/2,-5/2)$ we get that
    $b > -20 $ and  $r(e') > 80$.
	Thus, in all three cases $e'$ intersects $\ell'_1$ ($b + r \geq 1$).

	Finally, consider the sub-case where $c(e)$ has positive $x$-coordinate and $c(d)$ has negative $y$-coordinate.
We have the following equations:
	\begin{enumerate}[leftmargin=1.5\parindent]
		\item $r(e') = \frac{|  \sqrt{3} b  -a -2 |}{2} $,
		\item $a^2 + b^2 = (r(e') + 1)^2$, and
		\item $(a - x(q^-))^2 + (b - y(q^-))^2= r(e')^2$.
	\end{enumerate}
	  Solving these equations when $q^-=(0,-17/10)$ we get that
    $b > -17 $ and $r(e') > 25$.
		Solving these equations when $q^-=(-1/2,-183/100)$ we get that
    $b > -64 $ and $r(e') > 80 $.
		Thus, in all three cases $e'$ intersects $\ell'_1$ ($b + r \geq 1$).
\old{Let $t_1$ be the vertical line through point $(1,0)$, $t_2$ be the vertical line through point $(-1,0)$,
and let $t_3$ be the line tangent to $d^*$ at a point above the $x$-axis
of positive slope of angle $30^{\circ}$ with the $x$-axis.
We distinguish between two cases according to the $x$-coordinate of $e$'s center.

Observe that if $c(e)$ has negative $x$-coordinate it is sufficient to show that a disk $e'$ whose center
has negative $x$-coordinates and $y$-coordinates below $q^-$ that is tangent to $t_1$, $d^*$, and has the point $q^-$ on its boundary intersect $\ell'_1$, since this implies that $e$  intersect $\ell'_1$.

Moreover, if $c(e)$ has positive $x$-coordinate it is sufficient to show that a disk $e$ whose center
has positive $x$-coordinates and $y$-coordinates below $q^-$ that is tangent to $t_2$, $d^*$, and has the point $q^-$ on its boundary intersect $\ell'_1$, since this implies that $e$  intersect $\ell'_1$.

Notice that in the case the center of $d$ has negative $y$-coordinate the angle between $\ell_2$ and $x$-axis is at most $30^{\circ}$,
otherwise $\ell_2$ will not intersect $d$ (the radius of $d$ is at most 2).
Thus, if $c(e)$ has positive $x$-coordinate and $d$ has negative $y$-coordinate it is sufficient to show that a disk $e'$ whose center
has positive $x$-coordinates and $y$-coordinates below $q^-$ that is tangent to $t_3$, $d^*$, and has the point $q^-$ on its boundary intersect $\ell'_1$.

Let $t$ be $t_1$ is  $c(e)$ has negative $x$-coordinate, else if $d$ has positive $y$-coordinate let $t$ be $t_2$,
otherwise let $t$ be $t_3$.
Let $e'$ be a disk whose center is below $q^-$ that is tangent to $t$, $d^*$, and has the point $q^-$ on its boundary, and let $(a,b)$ be its center.

First consider the case where $c(e)$ has negative $x$-coordinate.

We have the following equations:
	\begin{enumerate}[leftmargin=1.5\parindent]
		\item $r = | 1 - a | $
		\item $a^2 + b^2 = (r(e') + 1)^2$, and
		\item $(a - x(q^-))^2 + (b - y(q^-))^2= r(e')^2$, and
	\end{enumerate}
	Solving these equations when $q^-=(0,-17/10)$ we get that
    $b > -5.4 $ and $r(e') > 7$.
	Solving these equations when $q^-=(1/2,-5/2)$ we get that
    $b > -4.9 $ and  $r(e') > 5.9$.
  Solving these equations when $q^-=(-1/2,-183/100)$ we get that
    $b > -13 $ and $r(e') > 40$.
		Thus, in all three cases $e'$ intersects $\ell'_1$ ($b + r \geq 1$).

Next, consider the case where $c(e)$ has positive $x$-coordinate and $d$ has positive $y$-coordinate,
We have the following equations:
	\begin{enumerate}[leftmargin=1.5\parindent]
		\item $r = | 1 + a | $
		\item $a^2 + b^2 = (r(e') + 1)^2$, and
		\item $(a - x(q^-))^2 + (b - y(q^-))^2= r(e')^2$, and
	\end{enumerate}
	Solving these equations when $q^-=(0,-17/10)$ we get that
    $b > -5.4 $ and $r(e') > 7$.
	Solving these equations when $q^-=(1/2,-5/2)$ we get that
    $b > -20 $ and  $r(e') > 80$.

	Finally, consider the case where $c(e)$ has positive $x$-coordinate and $d$ has negative $y$-coordinate,
We have the following equations:
	\begin{enumerate}[leftmargin=1.5\parindent]
		\item $r = \frac{|  \sqrt{3} b  -a -2 |}{2} $
		\item $a^2 + b^2 = (r(e') + 1)^2$, and
		\item $(a - x(q^-))^2 + (b - y(q^-))^2= r(e')^2$, and
	\end{enumerate}
		Solving these equations when $q^-=(-1/2,-183/100)$ we get that
    $b > -64 $ and $r(e') > 80 $.

		Thus, in all three cases $e'$ intersects $\ell'_1$ ($b + r \geq 1$).
}
\end{proof}

In the rest of this section we assume that the center of $e$ has negative $y$-coordinate and lies above  $q^-$.
Recall that we are in the {\color{mycolor}base setting} and the angle between $\ell_2$ and $\ell_3$ is the largest in the triangle $\Delta$ (formed by $\ell_1$, $\ell_2$, and $\ell_3$).
Let $\beta$ be the angle of $\Delta$ between $\ell_1$ and $\ell_3$, and let $\gamma$ be the angle of $\Delta$ between $\ell_1$ and $\ell_2$.
We distinguish between three cases according to the angle $\beta$. In each of these cases we show that $e$ does intersect both $\ell_2$ and $\ell_3$, and therefore cannot exist (in $D$).

\begin{figure}[hbt]
	\centering
	\includegraphics[scale=0.7]{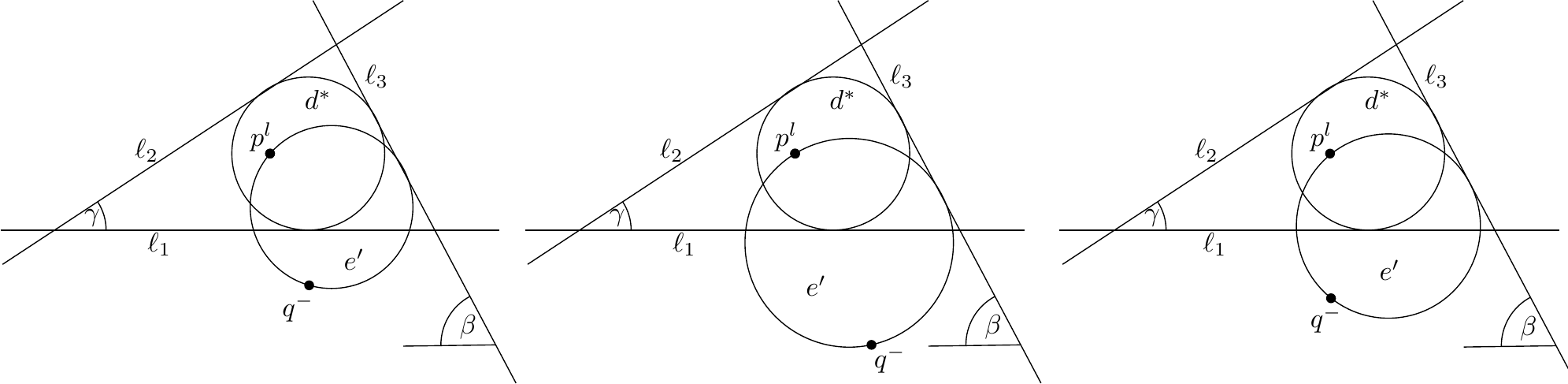}
	\caption{Left: $q^- = (0,-1.7)$. Middle: $q^- = (0.5,-2.5)$. Right: $q^- = (-0.5,-1.83)$.}
	\label{fig:disk_e_prime}
\end{figure}

Let $e'$ be the disk, with center of negative $y$-coordinate and above $q^-$, tangent to $p^l$, $q^-$ and $\ell_3$, where the point of tangency with $\ell_3$ is below $\ell'_1$ (otherwise, $e'$ clearly intersects $\ell'_1$).
We observe that it is sufficient to show that $e'$ does not intersect both $\ell_2$ and $\ell_3$, since this implies that any other disk whose center is of negative $y$-coordinate and above $q^-$
that avoids $p^l$ and $q^-$ and intersects $\ell_3$ below $\ell'_1$ does not intersect $\ell_2$; see Figure~\ref{fig:disk_e_prime}.

\begin{claim}
If the angle $\beta \leq 45^{\circ}$, then $e'$ does not intersect $\ell_2$.
\end{claim}
\begin{proof}
Let $t$ be the vertical line through the point $(-1,0)$.
Observe that it is sufficient to show that $e'$ does not intersect $t$,  since this implies that $e'$ does not  intersect $\ell_2$.
It is enough to show that $e'$ does not intersect $t$ when  $\beta = 45^{\circ}$.
Notice that if  $\beta = 45^{\circ}$, then the equation of $\ell_2$ is $y = -x + \sqrt{2}$.
Let $c(e')=(a,b)$.

Since $e'$ is tangent to $\ell_3$ and has $p^l$ and $q^-$ on its boundary, the following equations hold:
	\begin{enumerate}[leftmargin=1.5\parindent]
		\item $(a - x(p^l))^2 + (b - y(p^l))^2 = r(e)^2$,
		\item $(a - x(q^-))^2 + (b - y(q^-))^2= r(e)^2$, and
		\item $\frac{(a + b - \sqrt{2})^2}{2} = r(e)^2$.
	\end{enumerate}
	Solving these equations when $p^l=(-1/2,0)$ and $q^-=(0,-17/10)$ we get that
    $a > 0.44 $ and  $r(e') < 1.15$.
	Solving these equations when $p^l=(-1/2,0)$ and $q^-=(1/2,-5/2)$ we get that
    $a > 0.45 $ and  $r(e') < 1.44$.
  Solving these equations when $p^l=(-1/2,0)$ and $q^-=(-1/2,-183/100)$ we get that
    $a > 0.45 $ and $r(e') < 1.4$.
		Thus, in all cases $e'$ does intersect $t$, and therefore does not intersect $\ell_2$.
	\end{proof}

\begin{claim}
If the angle $45^{\circ} < \beta \le 75^{\circ}$, then $e'$ does not intersect $\ell_2$.
\end{claim}
\begin{proof}
Let $t$ be a line of positive slope that is tangent to $d^*$ at a point above the $x$-axis and forms an angle of $67.5^{\circ}$ with the $x$-axis; the equation of $t$ is $y = \tan(67.5^{\circ})x + \frac{1}{\cos(67.5^{\circ})}$.
Notice that $\gamma$, the angle between $\ell_2$ and the $x$-axis, is less than $67.5^{\circ}$ (since $\beta > 45^{\circ}$ and the largest angle in $\Delta$ is between $\ell_2$ and $\ell_3$).
Observe that it is sufficient to show that $e'$ does not intersect $t$,  since this implies that $e'$ does not intersect $\ell_2$.
It is enough to show that $e'$ does not intersect $t$ when  $\beta = 75^{\circ}$, and in this case the equation of $\ell_3$ is  $y = -(2 + \sqrt{3})x + \sqrt{2} + \sqrt{6}$.

Let $c(e')=(a,b)$.
To show that $e'$ does not intersect $t$ we need to show that the distance of the center of $e'$ from $t$ is more than $r(e')$, that is:
\[ (*) \qquad | -a \sin (67.5^{\circ}) + b \cos (67.5^{\circ}) - 1  | > r(e')\, . \]

Since $e'$ is tangent to $\ell_3$ and has $p^l$ and $q^-$ on its boundary, the following equations hold:
	\begin{enumerate}[leftmargin=1.5\parindent]
		\item $(a - x(p^l))^2 +(b - y(p^l))^2 = r(e')^2$,
		\item $(a - x(q^-))^2 + (b - y(q^-))^2= r(e')^2$, and
		\item $\frac{ |(2 + \sqrt{3})a + b -(\sqrt{2} + \sqrt{6})| }{\sqrt{(2+\sqrt{3})^2 +1 }} =  r(e')$.
	\end{enumerate}

	Solving these equations when $p^l=(-1/2,0)$ and $q^-=(0,-17/10)$ we get that
    $a > 0.19 $, $b < -0.71$ and  $r(e') < 1$.
	Solving these equations when $p^l=(-1/2,0)$ and $q^-= (1/2,-5/2)$ we get that
    $a  < -0.02 $,  $b > 1.26$ and  $r(e') < 1.346$.
  Solving these equations when $p^l=(-1/2,0)$ and $q^-= (-1/2,-183/100)$ we get that
    $a < 0.131 $,  $b > -0.92 $ and $r(e') < 1.111$.

Thus, by setting these values in equation $(*)$ we get that in all cases $e'$ does not intersect $t$,
and therefore does not intersect $\ell_2$.
\end{proof}

\begin{claim}
If the angle $75^{\circ} < \beta < 90^{\circ}$, then $e'$ does not intersect $\ell_2$.
\end{claim}
\begin{proof}
Let $t$ be a line of positive slope that is tangent to $d^*$ at a point above the $x$-axis and forms an angle of $30^{\circ}$ with the $x$-axis; the equation of $t$ is $y =  \frac{1}{\sqrt{3}}x + \frac{2}{\sqrt{3}}$.
Notice that $\gamma$, the angle between $\ell_2$ and the $x$-axis, is less than $30^{\circ}$ (since $\beta > 75^{\circ}$ and the largest angle in $\Delta$ is between $\ell_2$ and $\ell_3$).
Observe that it is sufficient to show that $e'$ does not intersect $t$,  since this implies that $e'$ does not intersect $\ell_2$.
It is enough to show that $e'$ does not intersect $t$ when  $\beta = 90^{\circ}$, and in this case the equation of $\ell_3$ is $x= 1$.

Let $c(e')=(a,b)$.
To show that $e'$ does not intersect $t$ we need to show that the distance of the center of $e'$ to $t$ is more than $r(e')$, that is:
\[ (*) \qquad \frac{| \sqrt{3}b - a - 2 |}{2} > r(e')\, .\]

Since $e'$ is tangent to $\ell_3$ and has $p^l$ and $q^-$ on its boundary, the following equations hold:
	\begin{enumerate}[leftmargin=1.5\parindent]
		\item $(a - x(p^l) )^2 +(b - y(p^l))^2 = r(e')^2 $,
		\item $(a - x(q^-) )^2 + (b - y(q^-))^2=  r(e')^2$, and
		\item $ |a - 1|      =  r(e')$.
	\end{enumerate}
	Solving these equations when $p^l=(-1/2,0)$ and $q^-=(0,-17/10)$ we get that
    $a = \frac{17\sqrt{471}}{25} - \frac{147}{10}$, $b = \frac{\sqrt{471}}{5} - \frac{51}{10}$ and  $r(e') = \frac{157}{10} - \frac{17\sqrt{471}}{25}$.
	Solving these equations when $p^l=(-1/2,0)$ and $q^-= (1/2,-5/2)$ we get that
    $a  < -0.42 $,  $b > -1.42$ and  $r(e') < 1.42$.
  Solving these equations when $p^l=(-1/2,0)$ and $q^-= (-1/2,-183/100)$ we get that
    $a = \frac{-1163}{40000}$,  $b = \frac{-183}{200} $ and $r(e') = \frac{-1163}{40000}$.

Thus, by setting these values in equation $(*)$ we get that in all cases $e'$ does not intersect $t$, and therefore does not intersect $\ell_2$.
\end{proof}

\old{
\begin{lemma}
\label{Lem:intersectEll1}
Let $e \in D$ be  a disk whose its center $y$-coordinates are below the $y$-coordinate of $q^-$ and does not contain
$q^-$ and intersects both  $\ell_2$ and $\ell_3$, then $e$ intersects $\ell'_1$.
\end{lemma}
\begin{proof}
To show this, we use the fact that $e$ intersects disks $d^*$ and $d$, and
that any such disk intersects $\ell'_1$.
Showing this is equivalent to show that any disk that is tangent to  $\ell_3$ and $\ell_2$,
having its center below $q^-$, and intersects $\ell'_1$ does not intersect both $d^*$ and $d$.

First consider the case where $q^- = (0,-1.7)$ (the case where $\alpha \leq 17^\circ$).
Assume w.l.o.g. that the center of $e$  is at the third quadrant of $q^-$.

Let $c$ be a disk center at the the third quadrant of $q^-$, such that it has $q^-$ on its boundary and it is
tangent to $\ell'_1$ and to the vertical line through point $(1,0)$.
Thus we have
\begin{enumerate}
    \item $r(c)= |-y(c) + 1|$, \ \ \ \ since   $c$ is tangent to $\ell'_1$.
		\item $r(c)= |-x(c) + 1|$, \ \ \ \ since   $c$ is tangent to the vertical line through point $(1,0)$.
		\item $(x(c) - x(q^-))^2 + (y(c) - y(q^-))^2=  r(c)^2$, \ \   having point $q^-$ on the boundary of $c$.
\end{enumerate}
These equations solve to
  \[ x(c) = - \frac{27}{10} -3 \sqrt{\frac{3}{5}}  , \ \ y(c)= - \frac{27}{10} -3 \sqrt{\frac{3}{5}}  \]

	By the above claim it suffices to show that $c$ does not intersect $d^*$, that is the distance between their centers is more than the sum of their radii:

		$  ( x(c) - 0 )^2 + (y(c) - 0)^2  >    (1 + r(c))^2$  
		This can be verified by plugging $x(c)$, $y(c)$, and $r(c)$, as follows:

	\begin{align}
	&\notag (x(c) - 0 )^2 + (y(c) - 0)^2 =\\
	&\notag \left(- \frac{27}{10} -3 \sqrt{\frac{3}{5}} \right)^2 + \left(- \frac{27}{10} -3 \sqrt{\frac{3}{5}} \right)^2 \ge\\
	&\notag \left(2 +  \frac{27}{10} + 3 \sqrt{\frac{3}{5}} \right)^2.
	\end{align}


	The above inequalities imply that $c$ does not intersect $d^*$. Therefore, since $e$ intersects $d^*$, $e$ must intersects $\ell'_1$.
	The proof for the case where the center of $e$  is at the fourth quadrant of $q^-$ is symmetric.


Next consider the case where $\alpha > 17^\circ$.
This case is split into two case with respect to the $x$-coordinate of $e$ (since $q^-$ is not on the $y$-axis,
the proof is not symmetric if $e$ is in the third or fourth quadrant of $q^-$).
Moreover, the proof of each sub-case is split  into two cases with respect to the two locations of $q^-$
(the two different cases with respect to the $y$-coordinates of $d$).

Assume that the center of $e$  is at the third quadrant of $q^-$.

Let $c$ be a disk centered at the the third quadrant of $q^-$, such that it has $q^-$ on its boundary and it is
tangent to $\ell'_1$ and to the vertical line through point $(1,0)$.
Thus we have
\begin{enumerate}
    \item $r(c)= |-y(c) + 1|$, \ \ \ \ since   $c$ is tangent to $\ell'_1$.
		\item $r(c)= |-x(c) + 1|$, \ \ \ \ since   $c$ is tangent to the vertical line through point $(1,0)$.
		\item $(x(c) - x(q^-))^2 + (y(c) - y(q^-))^2=  r(c)^2$, \ \   having point $q^-$ on the boundary of $c$.
\end{enumerate}
These equations solve to

$ x(c) =  -3 - \sqrt{\frac{7}{2}}  , \ \ y(c)=  - 3- \sqrt{\frac{7}{2}}  $
	when $q^- = (0.5,-2.5)$ and \\
   $ x(c) = - \frac{333}{100} - \frac{\sqrt{849}}{10} , \ \
	   y(c)=  - \frac{333}{100} - \frac{\sqrt{849}}{10} $  when $q^- = (-0.5,-1.83)$

	Now, we show that $c$ does not intersect $d^*$, that is the distance between their centers is more than the sum of their radii:

		$  ( x(c) - 0 )^2 + (y(c) - 0)^2  >    (1 + r(c))^2$  
		This can be verified by plugging $x(c)$, $y(c)$, and $r(c)$, as follows:

	\begin{align}
	&\notag (x(c) - 0 )^2 + (y(c) - 0)^2 =\\
	&\notag \left( -3 - \sqrt{\frac{7}{2}} \right)^2 + \left(-3 - \sqrt{\frac{7}{2}} \right)^2 \ge\\
	&\notag \left(2 +  3 + \sqrt{\frac{7}{2}} \right)^2.
	\end{align}

	and

	\begin{align}
	&\notag (x(c) - 0 )^2 + (y(c) - 0)^2 =\\
	&\notag \left( - \frac{333}{100} - \frac{\sqrt{849}}{10} \right)^2 + \left(- \frac{333}{100} - \frac{\sqrt{849}}{10} \right)^2 \ge\\
	&\notag \left(2 +   \frac{333}{100} + \frac{\sqrt{849}}{10} \right)^2.
	\end{align}


	The above inequalities imply that $c$ does not intersect $d^*$. Therefore, since $e$ intersects $d^*$, $e$ must intersects $\ell'_1$.

Finally, assume that the center of $e$  is at the fourth quadrant of $q^-$.

Let $\ell$ be the vertical line through point $(-1,0)$, and let $\ell'$ be the line that is tangent to $d^*$ and of convex angle
$30^{\circ}$ with line $\ell_1$.
Notice that in the case the center of $d$ has negative $y$-coordinate the angle between $\ell_2$ and $\ell_1$ is at most $30^circ$,
otherwise $\ell_2$ will not intersect $d$ (the radius of $d$ is at most 2).

Let $c$ be a disk centered at the the fourth quadrant of $q^-$, such that it has $q^-$ on its boundary and it is
tangent to $\ell'_1$ and to $\ell$, in the case that the $y$-coordinates of $d$ are positive.
And in the in the case that the $y$-coordinate of $d$ are negative,
let $c$ be a disk centered at the the fourth quadrant of $q^-$, such that it has $q^-$ on its boundary and it is
tangent to $\ell'_1$ and to $\ell'$.

Thus, we have
\begin{enumerate}
    \item $r(c)= |-y(c) + 1|$, \ \ \ \ since   $c$ is tangent to $\ell'_1$.
		\item $r(c)= |x(c) + 1|$, \ \ \ \ since   $c$ is tangent to $\ell$ (the vertical line through point $(-1,0)$).
		\item $(x(c) - x(q^-))^2 + (y(c) - y(q^-))^2=  r(c)^2$, \ \   having point $q^-$ on the boundary of $c$.
\end{enumerate}
These equations, when the $y$-coordinates of $d$ are positive and hence $q^-= (0.5,-2.5)$,
solve to

$ x(c) =        4 + \sqrt{\frac{21}{2}}  , \ \
  y(c)=  - 4 - \sqrt{\frac{21}{2}}  $ \\
	and when the $y$-coordinates of $d$ are negative and hence $q^-= (-0.5,-1.83)$, solve to \\
   $ x(c) =   \frac{233}{100} + \frac{\sqrt{283}}{10} , \ \
	   y(c)=  - \frac{233}{100} - \frac{\sqrt{283}}{10} $.

	Now, we show that $c$ does not intersect $d^*$, that is the distance between their centers is more than the sum of their radii:

		$  ( x(c) - 0 )^2 + (y(c) - 0)^2  >    (1 + r(c))^2$.
		This can be verified by plugging $x(c)$, $y(c)$, and $r(c)$, as follows:

	\begin{align}
	&\notag (x(c) - 0 )^2 + (y(c) - 0)^2 =\\
	&\notag \left(  4 + \sqrt{\frac{21}{2}}  \right)^2 +
	        \left( 4 + \sqrt{\frac{21}{2}}   \right)^2 \ge\\
	&\notag \left(2 +   4 + \sqrt{\frac{21}{2}}  \right)^2.
	\end{align}
and
\begin{align}
	&\notag (x(c) - 0 )^2 + (y(c) - 0)^2 =\\
	&\notag \left(  \frac{233}{100} + \frac{\sqrt{283}}{10}  \right)^2 +
	        \left( \frac{233}{100} + \frac{\sqrt{283}}{10}   \right)^2 \ge\\
	&\notag \left(2 +  \frac{233}{100} + \frac{\sqrt{283}}{10}  \right)^2.
	\end{align}

	The above inequalities imply that $c$ does not intersect $d^*$. Therefore, since $e$ intersects $d^*$, $e$ must intersects $\ell'_1$.
\end{proof}

} 


\section{Proofs of Claims~\ref{cl:R5moreThan02} and \ref{cl:R20lessThan05}}
\label{sec:app_B}
\setcounter{theorem}{9}
\begin{claim}
	Assume that $D^{-}_{ \leq 20} $ does not contain any disk with $\delta(\cdot)$ greater or equal to $0.5$.
	Let $d'$ be the disk in $D^{-}_{ \leq 5}$ with maximum $\delta(d')$, and assume that its center lies on the positive $x$-axis.
	If $0.11  \le \delta(d') < 0.5$, then $\{(0,0), (2,0), (-0.15, 1.75), (-0.15, -1.75)\}$ stabs $D$.
\end{claim}

\begin{proof}

	Any disk in $D\setminus D^-$ contains $(0,0)$. Thus, we prove the claim for $D^-$. We claim for any disk $e \in D^{-}$ that if $y(e)\ge 0$ then $e$ contains
  $(2,0)$ or $(-0.15, 1.75)$, otherwise $e$ contains
	$(2,0)$ or $(-0.15, -1.75)$.
	We prove our claim for $y(e)\ge 0$, since the proof for $y(e) < 0$ is symmetric. Notice that the center of $d'$ is $(r(d') + \delta(d'),0)$. We consider the following three cases.

	\paragraph*{$x(e) \le 0$.}
		Let $d_{2}$ be the disk of radius 2, with negative $x(d_2)$, that has $(0, 0)$ and $(-0.15, 1.75)$ on its boundary.
		If $(x_2,y_2)$ denotes the center of $d_{2}$, then we have the following equations
		\[ \qquad  x_{2}^2 + y_{2}^2 = 4, \qquad  (x_2 + 0.15)^2 +(y_2 - 1.75)^2 = 4,\]
		which solve to  $x_2 = - \frac{3}{40} - \frac{21 \sqrt{ {287}/{617}}}{8}$ and $y_2 = \frac{7}{8} - \frac{ 9 \sqrt{ {287}/{617}}}{40}$.

		Since $r(d') \leq 5$ and $ \delta(d') > 0.11$, the following inequality holds, and thus $d_2$ and $d'$ do not intersect: $(x_2 - (r(d') + \delta(d')) )^2 + (y_2 - 0)^2 > (2+ r(d'))^2$.

			As above, we use this fact to show that any disk $e \in D^-$ with $y(e) \ge 0$ and $x(e) \leq 0$ (that
		intersects both $d'$ and $d^*$) must contain the point $q= (-0.15,1.75)$.

		Let $d''$ be the disk of radius 5 centered at $(5.11,0)$. Notice that since $e$ intersects $d'$ it must also intersect $d''$.
		Let $t'$ be the (relevant) mutual tangent to $d''$ and $d^*$, i.e., $t': y = \frac{400}{\sqrt{101121}}x + \frac{511}{\sqrt{101121}}$. It is easy to verify that the point $q$ lies above $t'$. Now, let $t$ be the line parallel to $t'$ that passes through $q$, let $t_{\bot}$ be the line perpendicular to $t$ that passes through $q$, and let $P=P(q,d^*)$ be the bisector between $q$ and $d^*$.

		We first observe that in the halfplane to the left of the $y$-axis, $t_{\bot}$ is above  $P$. This can be verified by showing that for any point $p$ on $t_{\bot}$ with $x(p) \le 0$, we have that the distance between $p$ and $q$ is less than the distance between $p$ and $d^*$ (which is the distance between $p$ and $c^*$ minus 1).

	 If $c(e)$ is above $t_{\bot}$ then clearly $e$ contains $q$ (since $e$ intersects $d^*$), so assume that $c(e)$ is below $t_{\bot}$. Our assumptions on $e$ ($e \in D^-$, $r(e) > 2$, $c(e)$ is in the second quadrant) imply that $c(e)$ is above $t$.  Now, if $e$ intersects $t$ at a point above $q$, then by the triangle inequality $e$ contains $q$. Otherwise, since $e$ intersects $d''$, $e$ must intersect the clockwise arc of $d_2$ from $c^*$ to $q$. But, then $e$ contains $q$ (since $r(e) > r(d_2) =2$ and $c^* \notin e$).

		\paragraph*{$x(e) > 0$ and $r(e) \ge 5$.}
		Let $d_{5}$ be the disk of radius 5, with positive $x(d_5)$, that has $(2, 0)$ and $(-0.15, 1.75)$ on its boundary. If $(x_5,y_5)$ denotes the center of $d_{5}$, then we have
		\[ \qquad  (x_{5} - 2)^2 + y_{5}^2 = 25, \qquad  (x_{5} + 0.15)^2 +(y_{5} - 1.75)^2 = 25,\]
		which solve to  $x_5 = \frac{37}{40} + \frac{ 7 \sqrt{ {18463}/{1537}}}{8}$ and $y_5 = \frac{7}{8} + \frac{ 43 \sqrt{ {18463}/{1537}}}{40}$.

		Since the distance between the centers of $d_5$ and $d^*$ is greater than the sum of their radii, i.e., $(x_5)^2 + (y_5 )^2 > (1 + 5)^2$,
		these disks do not intersect.
		Thus, using observation~\ref{outside-strip} and~\ref{radius} as above, we have that $e$ contains point $(2,0)$ or point $(-0.15, 1.75)$.

		\paragraph*{$x(e) > 0$ and $r(e)<5$.}
		Recall that $r(e)> 2$. Since $e\in D^-_{\le 5}$, by the claim's assumption we have that $\delta(e) < 0.5$.
		Let $d_{2}$ be the disk of radius 2, with positive $x(d_2)$, that has $(2, 0)$ and $(-0.15, 1.75)$ on its boundary.	If $(x_2,y_2)$ denotes the center of $d_{2}$, then
		\[ \qquad  (x_{2} - 2)^2 + y_{2}^2 = 4, \qquad  (x_{2} + 0.15)^2 +(y_{2} - 1.75)^2 = 4,\]
		which solve to  $x_2 = \frac{37}{40} + \frac{ 7 \sqrt{ {1663}/{1537}}}{8}$ and $y_2 = \frac{7}{8} + \frac{ 43 \sqrt{ {1663}/{1537}}}{40}$.

		We claim that $d_{2}$ does not intersect the disk of radius $0.5$ and center $c^*$, which we denote by $d_{0.5}$.
		This can be verified by showing that the distance between the centers of these disks is greater than the sum of their radii, that is, $x_{2}^2 + y_{2}^2 > (0.5 + 2)^2$.
		Notice that $e$ intersects $d_{0.5}$ because $\delta(e) < 0.5$.
		Thus, again, using observations~\ref{outside-strip} and~\ref{radius} we have that $e$ contains point $(2,0)$ or point $(-0.15, 1.75)$. 
    %
	%
\end{proof}


Claims \ref{cl:R5moreThanHalf}, \ref{cl:moreThanHalf}, and \ref{cl:R5moreThan02} imply that if $D^-_{\le 5}$ contains a disk with $\delta(\cdot) \ge 0.11$, or if $D^-_{\le 20}$ contains a disk with $\delta(\cdot) \ge 0.5$, then there exists a set of four points that stabs $D$. It remains to prove our claim for the case where every disk in $D^-_{\le 5}$ has $\delta(\cdot) < 0.11$ and every disk in $D^-_{\le 20}$ has $\delta(\cdot) < 0.5$. To that end, we assume the {\color{mycolor} base setting}.


In what follows we show that in this case $D$ is stabbed by $(0,0),(2.5,1),(-2.5,1),$ and $(0,-1.52)$.
Here we assume w.l.o.g. that  one of the tangents ($\ell_1$) is as described in Section~\ref{secDminBig}. See Figure~\ref{fig:finA}.

\begin{figure}[hbt]
	\centering
	\includegraphics[scale=0.4]{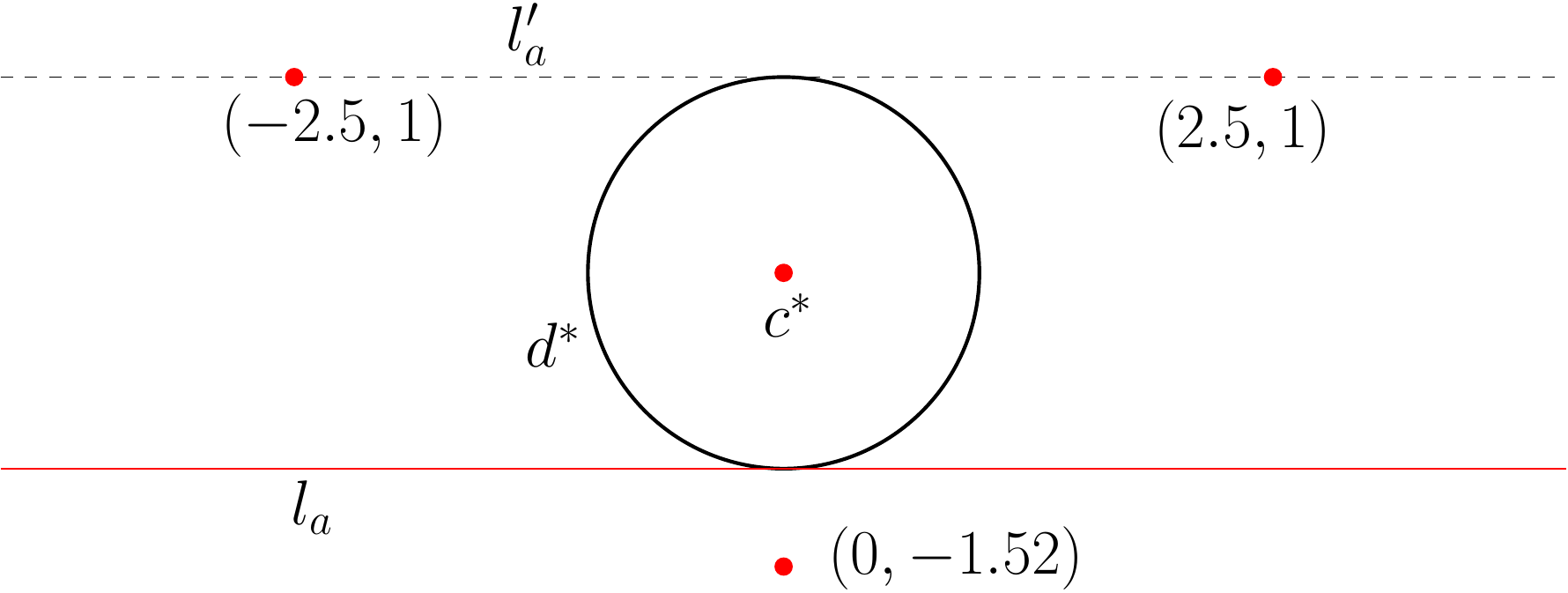}
	\caption{The 4 piercing points in this case are depicted in red.}
	\label{fig:finA}
\end{figure}

\begin{claim}
	If for every disk  $e_1 \in D^{-}_{ \leq 5}$ we have $\delta(e_1) < 0.11$ and for every disk $e_2 \in D^{-}_{ \leq 20}$ we have $\delta(e_2) < 0.5$, then
	the set $\{(0,0),(2.5,1),(-2.5,1),(0,-1.52)\}$ stabs $D$.
\end{claim}
\begin{proof}
  We prove that any disk $e \in D^-$ whose center has positive $x$-coordinate is stabbed by at least one of the points $(2.5,1)$ or $(0,-1.52)$.
	The case where $e$'s center has negative  $x$-coordinate is symmetric.

	Let $d_{k}$  be the disk of radius $k$ having points $(0,-1.52)$ and $(2.5,1)$ on its boundary and its center is below the line $l$ that passes through points
	$(0,-1.52)$ and $(2.5,1)$. Let $V$ be the strip defined by the two lines perpendicular to $l$ and passing through the points $(2.5,1)$ and $(0,-1.52)$.

	Consider the disk $d_{2}$ (i.e., $k=2$), whose center $(x_2 , y_2)$ is below $l$ that is obtained by the following equations
		\[    (x_{2} - 2.5)^2 + (y_{2} - 1)^2 = 4, \qquad  x_{2}^2 +(y_{2} + 1.52)^2 = 4,    \qquad (c(e) \text{ below } l ) \]
		which solve to
		$x_2 = \frac{5}{4} + \frac{63 \sqrt{\frac{8499}{109}}}{850}$
		and $y_2 = \frac{-13}{50} - \frac{5 \sqrt{\frac{8499}{109}}}{68} $.

	We have that $\delta(d_2) = \sqrt{ x_2^2 + y_2^2 } -2 > 0.11$, that is disk $d_2$ does not intersect the disk
	$d^*_{0.11}$ of radius 0.11 centered at the origin.
 By our assumptions any disk $e \in D^{-}_{ \leq 5}$ intersects   $d^*_{0.11}$.
By Observations~\ref{outside-strip} and~\ref{radius}, we have that $e$ (recall that all disks in $D^-$ have radius at least 2) contains at least one of the points $(0,-1.52)$ or  $(2.5,1)$,
where  $d^*_{0.11}$, $e$,         $(2.5,1)$ and $(0,-1.52)$ play the role of
       $\eta$,        $\epsilon$,   $a$       and        $b$, respectively.

	Moreover, we have that $\delta(d_5) > 0.5$ (this can be verified as above by computing the location of the center of $d_{5}$ and its distance to the origin),
	that is disk $d_5$ does not intersect the disk
	$d^*_{0.5}$ of radius 0.5 centered at the origin.
 By our assumptions any disk $e \in D^{-}_{ \leq 20}$ intersects   $d^*_{0.5}$.
 By Observations~\ref{outside-strip} and~\ref{radius}, we have that $e$ with $r(e) \geq 5$ contains at least one of the  points $(0,-1.52)$ or $(2.5,1)$,
  where  $d^*_{0.5}$, $e$,         $(2.5,1)$ and $(0,-1.52)$ play the role of
       $\eta$,        $\epsilon$,   $a$       and        $b$, respectively.

	Finally, we have that $d_{20}$  does not intersect $d^*$ (this can be verified as above).
 By Observations~\ref{outside-strip} and~\ref{radius}, we have that $e \in D^-$ with $r(e) \geq 20$ contains at least one of the points $(0,-1.52)$ or  $(2.5,1)$ (since $e$ intersects $d^*$),
  where  $d^*$, $e$,         $(2.5,1)$ and $(0,-1.52)$ play the role of
       $\eta$,        $\epsilon$,   $a$       and        $b$, respectively.

	For a disk $e \in D^-$ centered above the line containing $(-2.5, 1)$ and $(2.5,1)$,
	the argument follows as the case where $\rmin \geq 4$.

	Finally, the union of disks of radius 2 centered at points
	$(0,0), \ (2.5,1), \ (-2.5,1),$ and $(0,-1.52)$ covers the region (below the line containing $(-2.5, 1)$ and $(2.5,1)$,
	above the line containing $(0, -1.52)$ and $(2.5,1)$, and
	above the line containing $(0, -1.52)$ and $(-2.5,1)$)
	that is not considered by the above cases.
\end{proof}

\end{document}